\documentclass[11pt,a4article]{article}
\usepackage[ruled,vlined]{algorithm2e}
\usepackage{t1enc}
\usepackage[mathscr]{eucal}
\usepackage{indentfirst}
\usepackage{graphicx}
\usepackage{graphics}
\usepackage{pict2e}
\usepackage{epic}
\usepackage[margin=2.9cm]{geometry}
\usepackage{epstopdf}
\usepackage{amsmath,amsthm,verbatim,amssymb,amsfonts,amscd, graphicx}
\usepackage{mathtools}
\usepackage[dvipsnames]{xcolor}

\usepackage{bm}
\usepackage{setspace}
\numberwithin{equation}{section}

\usepackage[T1]{fontenc}
\usepackage[utf8]{inputenc}
\usepackage{authblk}

\theoremstyle{plain}
\newtheorem{theorem}{Theorem}[section]

\newtheorem{lemma}{Lemma}[section]

\theoremstyle{definition}
\newtheorem{definition}{Definition}[section]
\newtheorem{remark}{Remark}[section]

\usepackage{algorithmic}
\usepackage[ruled,vlined]{algorithm2e}
\usepackage{multirow}

\def\co{{\text{conv}}}
\def\R{{\mathbb R}}
\def\Q{{\bm Q}}
\def\P{{\mathbb P}}

\def\LL{{\bm L}}
\def\N{{\mathbb N}}

\def\D{{\bm D}}
\def\U{{\bm U}}

\def\V{{\bm V}}
\def\K{{\bm K}}

\def\X{{\bm X}}
\def\A{{\bm A}}
\def\B{{\bm B}}
\def\S{{\bm S}}

\def\cs{{\text{cs}}}
\def\ex{{\text{ex}}}
\def\rank{{\text{rank}}}
\def\emd{{\text{emd}}}
\def\opt{{\text{opt}}}
\newcommand{\E}{\mathbb{E}}
\renewcommand{\epsilon}{\varepsilon}
\newcommand*\samethanks[1][\value{footnote}]{\footnotemark[#1]}

\DeclareMathOperator*{\argmin}{arg\,min}

\usepackage[backend=bibtex,date=year,giveninits=true,sorting=nyt,natbib=true,maxcitenames=2,maxbibnames=10,url=false,doi=true,backref=false]{biblatex}
\renewbibmacro{in:}{\ifentrytype{article}{}{\printtext{\bibstring{in}\intitlepunct}}}

\usepackage[colorlinks,linkcolor = BrickRed, citecolor=blue]{hyperref}

\definecolor{mypink1}{rgb}{0.858, 0.188, 0.478}
\definecolor{mypink2}{RGB}{219, 48, 122}
\definecolor{mypink2}{cmyk}{0, 0.7808, 0.4429, 0.1412}
\definecolor{mygray}{gray}{0.6}

\begin{document}
	
	\title{Probabilistic methods for approximate archetypal analysis}
	\date{}
	
	\author[1]{\small Ruijian Han\thanks{Equal contribution.}}
	\author[2]{\small Braxton Osting}
	\author[3,4]{\small Dong Wang}
	\author[2,5]{\small Yiming Xu\samethanks}
	
	\affil[1]{\footnotesize Department of Statistics, The Chinese University of Hong Kong, Hong Kong, China}
	\affil[2]{\footnotesize Department of Mathematics, University of Utah, Salt Lake City}
	\affil[3]{\footnotesize School of Science and Engineering, The Chinese University of Hong Kong, Shenzhen}
	\affil[4]{\footnotesize Guangdong Provincial Key Laboratory of Big Data Computing, The Chinese University of Hong Kong, Shenzhen}
	\affil[5]{\footnotesize Scientific Computing and Imaging Institute, University of Utah, Salt Lake City}

	\maketitle
	
	\begin{abstract}
		Archetypal analysis is an unsupervised learning method for exploratory data analysis. One major challenge that limits the applicability of archetypal analysis in practice is the inherent computational complexity of the existing algorithms. In this paper, we provide a novel approximation approach to partially address this issue. Utilizing probabilistic ideas from high-dimensional geometry, we introduce two preprocessing techniques to reduce the dimension and representation cardinality of the data, respectively. We prove that provided the data is approximately embedded in a low-dimensional linear subspace and the convex hull of the corresponding representations is well approximated by a polytope with a few vertices, our method can effectively reduce the scaling of archetypal analysis. Moreover, the solution of the reduced problem is near-optimal in terms of prediction errors. Our approach can be combined with other acceleration techniques to further mitigate the intrinsic complexity of archetypal analysis. We demonstrate the usefulness of our results by applying our method to summarize several moderately large-scale datasets. 
	\end{abstract}
	

	\section{Introduction}
	Archetypal analysis (AA) is an unsupervised learning method introduced by Cutler and Breiman in 1994 \cite{cutler1994archetypal}. For fixed $k\in\mathbb N$, the method finds a convex polytope with $k$ vertices, referred to as \emph{archetypes}, in the convex hull of the data that explains the most variation of the data. 
	Equivalently, given $\{x_i\}_{i\in [N]}\subset\R^d$, AA can be formulated as the following optimization problem:
	\begin{align}
		&\min_{\A\in\R_{\cs}^{N\times k},\B\in\R_\cs^{k\times N}} 
		\frac{1}{\sqrt{N}} \|\X-\X\A\B\|_F 
		&\X = [x_1, \cdots, x_N]\in\R^{d\times N},
		\label{1}
	\end{align}
	where $F$ denotes the Frobenius norm, and `$\cs$' stands for column stochastic matrices, which are entry-wise nonnegative matrices with each column summing to 1.
	The normalizing factor $1/\sqrt{N}$ is introduced for convenience later. 
	To understand this formulation, note that the columns of $\X\A$ are the expected archetypes, and the columns of $\B$ correspond to the projection coefficients of the columns of $\X$ to the convex hull of the archetypes. Consequently, the objective defined in \eqref{1} represents the (average) variation of the data that cannot be explained by the convex combinations of the archetypes. 
	
	AA is closely related to other unsupervised learning methods such as the $k$-means, principal component analysis (PCA) and nonnegative matrix factorization (NMF) \cite{hastie01statisticallearning, javadi2020nonnegative}. 
	In fact, AA can be seen as an interpolation between the $k$-means and PCA; it has more geometry than the former while it is more restrictive than the latter due to additional convexity constraints.
	This allows AA to produce more interpretable results in many applications, e.g., in evolutionary biology \cite{shoval2012evolutionary}, meanwhile raising additional questions of increased computational complexity. 
	Under suitable assumptions, the consistency and convergence of AA have recently been established in \cite{osting2021consistency}, laying the foundation for AA to be applicable to large-scale inference.   
	
	Despite offering interpretable results, AA did not gain equal attention compared to its alternatives. 
	One possible reason, as pointed out in \cite{chen2014fast}, is due to the lack of efficient computational resources for applying AA to large-scale datasets, which are becoming increasingly ubiquitous in the big-data era. 
	Indeed, the optimization defined in \eqref{1} is non-convex, and one common approach to solving \eqref{1} is based on an alternating minimization algorithm \cite{cutler1994archetypal}, which will be reviewed in Section \ref{sec:2}.  
	The subproblems in the alternating minimization scheme are equivalent to quadratic programming problems (see Section \ref{sec:2}), which makes the full loop for solving AA computationally intensive for moderately large dimension $d$ and cardinality $N$.

	The scope of this paper is to provide a promising perspective for addressing the theoretical computational challenges encountered by the AA. 
	Instead of focusing on optimizing the subproblem solvers to accelerate computation, we introduce two separate dimensionality reduction techniques to downsize the problem before applying optimization methods to solve  \eqref{1}.  
	We show that under appropriate conditions, a solution of the reduced AA 
	(i) well-approximates the solution of the original problem \eqref{1} in terms of  projection error and
	(ii) can be obtained significantly faster than the original solution. 
	Our approach relies on a few fundamental results in high-dimensional geometry.  
	Note that our proposed method is a data preprocessing procedure by nature, and complements the many existing methods to further accelerate computation.  
	
	\subsection{Related work}
	
	Making archetypal analysis practical for large-scale data analysis has been an active area of research in recent years. 
	Various approaches have been proposed to attack the problem from different perspectives. 
	For example, feasible optimization techniques such as projected gradients \cite{morup2012archetypal}, active-subsets \cite{chen2014fast}, and the Frank-Wolfe method  \cite{bauckhage2015archetypal} are considered for accelerating solving the quadratic programming problem in the alternating minimization scheme. 
	Relaxation methods including decoupling \cite{mei2018online} and sparse projections \cite{abrol2020geometric} are concerned with relaxing the alternating minimization into problems that enjoy better scalability properties. 
	Another direction of work is centered around approximately solving AA by first reducing the cardinality of the data via sparse representation \cite{thurau2011convex, mair2017frame}. 
	Although these approaches are demonstrated to work well empirically, they either do not address the intrinsic complexity of the problem or lack theoretical guarantee on the quality of approximation. 
	In the recent work \cite{mair2019coresets}, the authors proposed to use the coreset of the data to reduce the computational complexity of the objective function and theoretically quantified the approximation error. 
	
	Using approximate isometric embedding to reduce dimensionality is a fruitful idea in data analysis. The technique has been successfully applied to a variety of problems including least-squares regression \cite{drineas2011faster, avron2010blendenpik}, clustering \cite{boutsidis2010random, cohen2015dimensionality, makarychev2019performance}, low-rank approximation \cite{tropp2017practical,clarkson2017low,halko2009finding}, nonnegative matrix factorization  \cite{erichson2018randomized, qian2018dsanls}, and tensor decomposition \cite{wang2015fast, battaglino2018practical}.
	
	\subsection{Contributions of this paper}
	This paper proposes two novel dimensionality reduction techniques which can be combined with existing approaches to mitigate the inherent complexity of archetypal analysis.
	Both techniques come with theoretical guarantees on their approximation accuracy. 
	In particular,
	\begin{itemize}
		\item We introduce a data compression technique based on a randomized Krylov subspace method \cite{musco2015randomized} to reduce data dimension.
		This procedure allows us to circumvent frequent queries to high-dimensional data and is new in the context of archetypal analysis.
		\item We propose to use random projections to compute an approximate convex hull of the data to reduce the cardinality of the dictionary to represent archetypes. 
		\item We theoretically analyze the approximation accuracy and time complexity for both techniques. In particular, we show that the reduced archetypal analysis gives a near-optimal solution but has significantly reduced complexity provided that the data is low-dimensional and approximately described by a few extreme patterns. 
	\end{itemize}
	Our results yield an approximate algorithm that is capable of dealing with data that is large both in size and dimension.
	Numerical experiments are provided which support and illustrate our theoretical findings. 
	
	\subsection{Outline}
	The rest of the paper is organized as follows. In Section \ref{sec:2}, we review the standard alternating minimization algorithm for solving archetypal analysis as well as the corresponding computational challenges.
	In Section \ref{sec:data} and \ref{sec:card}, we introduce two separate randomized techniques to reduce the data dimension and representation cardinality of the archetypes, respectively. 
	We also quantify the approximation accuracy and the computational complexity for both techniques.  
	In Section \ref{sec:comb}, we combine the ideas in Section \ref{sec:data} and \ref{sec:card} to devise an approximate algorithm for archetypal analysis.
	We show that the proposed algorithm gives a near-optimal solution meanwhile having significantly reduced computational complexity for datasets that are approximately embedded in a low-dimensional subspace and well summarized via a few extreme points. We numerically verify our results in Section \ref{sec:num}.    
	
	\subsection{Notation}
	In the rest of the paper, $\X \in \R^{d\times N}$ denotes the data matrix. 
	We always use $(\A_\star, \B_\star)$ to denote a minimizer to \eqref{1}, and $\opt(\X) = \|\X-\X\A_\star \B_\star \|_F/\sqrt{N} $ the corresponding optimum value. 
	
	Denote $[m] = \{ 1, \ldots, m \} \subset \N$. 
	For a matrix $\A\in\R^{m\times n}$, we denote by $\sigma_i(\A)$ the $i$-th largest singular value of $\A$, and $\A^\dagger$ the Moore-Penrose pseudoinverse of $\A$.   
	For $T_1\subseteq [m]$ and $T_2\subseteq [n]$, we use notation 
	$\A[T_1, :]$,  
	$\A[-T_1, :]$, 
	$\A[:, T_2]$, and 
	$\A[:, -T_2]$ to denote the submatrices formed by taking the rows of $\A$ with indices in $T_1$, the rows of $\A$ with indices in $[m]\setminus T_1$, the columns of $\A$ with indices in $T_2$, and the columns of $\A$ with indices in $[n]\setminus T_2$, respectively. 
	When talking about subspace embedding for $\A$, we view $\A$ as $n$ points $\A[:, 1], \cdots, \A[:, n]$ in the column space of $\A$, i.e., $\text{col}(\A)$.
	We use $\co(\A)$ and $\ex(\A)$ to represent the convex hull of the columns of $\A$ and the corresponding extreme points, respectively. 
	
	Moreover, $\mathcal O(\cdot )$, $a(n_1, \cdots, n_\ell)\lesssim b(n_1, \cdots, n_\ell)$ and $a(n_1, \cdots, n_\ell)\gtrsim b(n_1, \cdots, n_\ell)$ are standard notation in complexity theory, where the implicit constants do not depend on the indices $n_1, \cdots, n_\ell$.

	\section{An alternating minimization algorithm for archetypal analysis}\label{sec:2}
	In this section, we review an alternating minimization algorithm for solving AA, due to Cutler and Breiman \cite{cutler1994archetypal}. 
	
	Note that \eqref{1} is a non-convex optimization. 
	However, when fixing $\A$ or $\B$ and solving for the other, the problem becomes convex. This observation gives rise to the following alternating minimization algorithm for computing a stationary solution for \eqref{1}.
	\medskip
	
	\begin{algorithm}[H]
		\KwIn{$\{x_i\}_{i\in [N]}$: dataset, $k$: number of archetypes}
		\KwOut{$\A$, $\B$}
		\begin{algorithmic}[1]
			\STATE {Initialize $\X\A$}
			\WHILE{not converged}{
				\STATE {$\B\gets\argmin_{\B'\in\R_\cs^{k\times N}}\|\X-\X\A\B'\|_F^2$}
			}
			\STATE {$\A\gets\argmin_{\A'\in\R_\cs^{N\times k}}\|\X-\X\A'\B\|_F^2$}
			\ENDWHILE
			\STATE {final update for $\B$:  $\B\gets\arg\min_{\B'\in\R_\cs^{k\times N}}\|\X-\X\A\B'\|_F^2$}
			\STATE{return $\A, \B$}
		\end{algorithmic}
		\caption{Alternating Minimization Algorithm for AA \cite{cutler1994archetypal}} 
		\label{alg:AM}
	\end{algorithm}
	\medskip
	
	The loop in Algorithm \ref{alg:AM} updates $\B$ and $\A$ alternatingly. 
	To analyze the computational complexity of these subroutines, we formulate the optimization problems in steps $3$ and $4$ more explicitly as follows. 
	
	In step 3, $\A$ is fixed and $\B$ needs to be updated. If we let $\bm Z = \X\A$, then the optimization is equivalent to computing the projection coefficients for each column in $\X$ to $\co(\bm Z)$. In particular, we need to solve $N$ independent $k$-dimensional quadratic programming problems:
	\begin{align*}
		&\min_{b\in\R^k,\|b\|_1 = 1, b\geq 0}\|\bm Z b-\X[:, i]\|_2^2& i\in [N]. 
	\end{align*}
	
	In step 4, $\B$ is fixed and $\A$, or equivalently, $\bm Z$, needs to be updated.
	Using the Pythagorean theorem, one can first compute the least-squares solutions
	\begin{align*}
	\argmin_{\bm Z\in\R^{d\times k}}\|\X-\bm Z\B\|_F^2 = ((\B^T)^\dagger \X^T)^T = \X\B^T(\B\B^T)^{-1},
	\end{align*}
	then update each column of $\A$ by projection:
	\begin{align}
	&\min_{a\in\R^N, \|a\|_1=1, a\geq 0}\left\|\X a - \X\B^T(\B\B^T)^{-1}[:, i]\right\|_2^2& i\in [k].\label{rev1}
	\end{align}
	Alternatively, one can use a Gauss-Seidel approach to update the columns of $\bm Z$ sequentially to accelerate computation \cite{osting2021consistency}. 
	Since the rest of the paper uses the Gauss-Seidel technique in the subroutine of solving reduced AA, we derive the optimization problems resulting from the procedure; more details can be found in \cite[Appendix B]{osting2021consistency}. 
	
	The Gauss-Seidel method updates the identified archetypes (i.e. the columns of $\bm Z$) one at a time. 
	In the $i$-th step, the procedure optimizes over the $i$-th column of $\bm Z$ with the rest kept fixed. It can be verified from direct computation that for $i\in [k]$, 
	\begin{align*}
	\|\X-\bm Z\B\|_F^2 &= \sum_{j\in [d]}\sum_{\ell\in [N]}\left[\X[j,\ell]^2-2\X[j,\ell]\sum_{s\in [k]}\bm Z[j,s]\B[s, \ell]+\left(\sum_{s\in [k]}\bm Z[j,s]\B[s, \ell]\right)^2\right]\\
	& = \sum_{j\in [d]}\sum_{\ell\in [N]}\left[(\bm Z[j,i]\B[i,\ell])^2-2\bm Z[j, i]\B[i,\ell]\left(\X[j,\ell]-\sum_{s\neq i}\bm Z[j,s]\B[s, \ell]\right)\right]+\Delta\\
	& = \|\B[i, :]\|^2_2\sum_{j\in [d]}\left[\bm Z[j,i]- \frac{1}{\|\B[i,:]\|^2_2}\sum_{\ell\in [N]}\B[i,\ell]\left(\X[j,\ell]-\sum_{s\neq i}\bm Z[j,s]\B[s, \ell]\right)\right]^2+\Delta\\
	& = \|\B[i, :]\|^2_2\left\|\bm Z[:, i] - \frac{\D_i(\B[i, :])^T}{\|\B[i, :]\|^2_2}\right\|_2^2+\Delta,
	\end{align*} 
	where $\D_i = \X - \bm Z[:,-i]\B[-i, :]$ and $\Delta$ collects the terms that do not depend on $\bm Z[:,i]$. 
	Since $\bm Z[:, i] = \X\A[:,i]$, minimizing $\|\X-\bm Z\B\|_F^2$ is equivalent to solving 
	\begin{align}
	&\min_{a\in\R^N, \|a\|_1=1, a\geq 0}\left\|\X a - \frac{\D_i(\B[i, :])^T}{\|\B[i, :]\|^2_2}\right\|_2^2& i\in [k].\label{rev2}
	\end{align}
	Either \eqref{rev1} or \eqref{rev2} involves solving $k$ quadratic programming problems with variable dimension $N$.
	
	For small $k$ and large $N$, the computation time in step 3 scales linearly in $N$ (assuming solving a $k$-dimensional quadratic programming problem takes constant time). For step 4, the computation time is approximately equal to a multiplicative constant ($k$) times the complexity of solving an $N$-dimensional quadratic programming problem, which can be computationally infeasible for large $N$. 
	We will provide a theoretically justified accelerated scheme for step 4 in Section \ref{sec:card}. 
	Moreover, when $d$ is large, taking repeated numerical operations on $\X$ is inconvenient. We will introduce a data dimensionality reduction technique to address this issue in Section \ref{sec:data}.

	\section{Data dimensionality reduction}\label{sec:data}
	
	We first consider the scenario where the data dimension is large. 
	This may happen, for instance, when each data point is obtained from the discretization of a continuous function (time series) or encodes a high-resolution image.  
	In this case, directly working with the data is inconvenient. 
	Instead, we can embed $\X$ in a lower dimensional space while maintaining the convexity structure of $\X$. 
	This compression will save us from frequently querying the columns of $\X$ in the iterative process for solving \eqref{1}, which can be computationally expensive. 
	A straightforward idea for embedding is via singular value decomposition (SVD), which we recall below: 
	
	\begin{definition}
		Suppose $\rank(\X) = r\leq \min\{N, d\}$. 
		The singular value decomposition (SVD) of $\X$ is given by $\X = \U\bm\Sigma\V^T$, where $\U\in\R^{d\times r}, \V\in\R^{r\times N}$ are the left and right singular vector matrices, respectively, and $\bm\Sigma\in\R^{r\times r}$ is a diagonal matrix with diagonal entries arranged in non-increasing order.  
	\end{definition} 
	
	Under the columns of $\U$, $\bm\Sigma\V^T\in\R^{r\times N}$ provides a sparse representation for $\X$ (since $r\leq d$). If we first embed $\X$ in $\U$ using SVD and apply AA to $\bm\Sigma\V^T$, then for every feasible $(\A, \B)$, by the unitary invariance of Frobenius norm, 
	\begin{align}
		\left\|\bm\Sigma\V^T-\bm\Sigma\V^T\A\B\right\|_F^2 = \left\|\U^T\X-\U^T\X\A\B\right\|_F^2 = \|\X-\X\A\B\|_F^2,\label{svd-aa}
	\end{align}
	which establishes the equivalence between \eqref{1} and the AA under the SVD representation. 
	
	In fact, if $\X$ has full rank but possesses low-rank structure, one may use a truncated SVD to further reduce the data dimension at a minor cost of accuracy, as made precise in the following theorem:
	
	\begin{theorem}\label{thm:tSVD}
		Suppose $p\leq r = \rank(\X)$. 
		Denote by $\U_p, \V_p$ the first $p$ columns of $\U$ and $\V$, respectively, and $\bm\Sigma_p$ the top $p\times p$ submatrix of $\mathbf \Sigma$.  
		Let $(\widetilde{\A}, \widetilde{\B})$ be a solution to the AA for the truncated SVD representation of $\X$ at $p$-th level:
		\begin{align*}
			\min_{\A\in\R_{\cs}^{N\times k},\B\in\R_\cs^{k\times N}}\frac{1}{\sqrt{N}}\left\|\mathbf \Sigma_p\V_p^T-\mathbf \Sigma_p\V_p^T\A\B\right\|_F.
		\end{align*}
		Then, 
		\begin{align}
			\frac{1}{\sqrt{N}}\left\|\X-\X\widetilde{\A}\widetilde{\B}\right\|_F\leq \opt(\X) + 4\sigma_{p+1}(\X).\label{truncate}
		\end{align}
	\end{theorem}
	\begin{proof}
		Let 
		\begin{align}
			&\X_p: = \U_p\bm\Sigma_p\V_p&\X_{-p}: = \X-\X_p.\label{Xp}
		\end{align}
		By the Eckart--Young theorem \cite{eckart1936approximation}, $\X_p$ is the best rank-$p$ approximation for $\X$ in the spectral norm, with approximation error $\|\X_{-p}\|_2 = \sigma_{p+1}(\X)$. 
		Let $(\A_\star, \B_\star)$ be a solution to \eqref{1}. 
		Consequently, 
		\begin{align}
			\|\X-\X\widetilde{\A}\widetilde{\B}\|_F &\leq  \|\X_p-\X_p\widetilde{\A}\widetilde{\B}\|_F + \|\X_{-p}-\X_{-p}\widetilde{\A}\widetilde{\B}\|_F\nonumber\\
			& \leq  \|\X_p-\X_p\A_\star\B_\star\|_F + \|\X_{-p}-\X_{-p}\widetilde{\A}\widetilde{\B}\|_F\nonumber\\
			&\leq \|\X-\X\A_\star\B_\star\|_F + \|\X_{-p}-\X_{-p}\A_\star\B_\star\|_F + \|\X_{-p}-\X_{-p}\widetilde{\A}\widetilde{\B}\|_F\nonumber\\
			&\leq \|\X-\X\A_\star\B_\star\|_F + 2\|\X_{-p}\|_F + \|\X_{-p}\A_\star\B_\star\|_F + \|\X_{-p}\widetilde{\A}\widetilde{\B}\|_F.\label{etc}
		\end{align}
		Since $\A_\star, \widetilde{\A}, \B_\star, \widetilde{\B}$ are column stochastic matrices,  so are $\A_\star\B_\star$ and $\widetilde{\A}\widetilde{\B}$. 
		It follows from direct computation and Cauchy-Schwarz inequality that
		\begin{align}
			\|\X_{-p}\A_\star\B_\star\|_F = \sqrt{\sum_{i\in [N]}\|\X_{-p}(\A_\star\B_\star)[:, i]\|_2^2}&\leq \sqrt{\sum_{i\in [N]}\|\X_{-p}\|^2_2\|(\A_\star\B_\star)[:, i]\|^2_2}\nonumber\\
			&\leq \sqrt{\sum_{i\in [N]}\|\X_{-p}\|^2_2\|(\A_\star\B_\star)[:, i]\|^2_1}\nonumber\\
			&\leq \|\X_{-p}\|_2\sqrt{N}.\label{>>>}
		\end{align}
		Similarly, 
		\begin{align}
			\|\X_{-p}\widetilde{\A}\widetilde{\B}\|_F\leq \|\X_{-p}\|_2\sqrt{N}.\label{125}
		\end{align} 
		Plugging \eqref{>>>} and \eqref{125} into \eqref{etc} and dividing by $\sqrt{N}$ yields
		\begin{align*}
			\frac{1}{\sqrt{N}}\|\X-\X\widetilde{\A}\widetilde{\B}\|_F\leq\opt(\X) + \frac{2}{\sqrt{N}}\|\X_{-p}\|_F + 2\|\X_{-p}\|_2&\leq\opt(\X) + 4\|\X_{-p}\|_2\\
			& = \opt(\X) + 4\sigma_{p+1}(\X), 
		\end{align*}
		completing the proof. 
	\end{proof}
	
	
	Thus, for data $\X$ that admits a good low-rank approximation, AA applied to the truncated SVD representation yields a near-optimal solution in terms of prediction errors.  
	In this case, the data dimension can be significantly reduced to streamline computation. 
	However, to obtain truncated SVD representations, one often needs to compute the full SVD of $\X$, which has complexity $\mathcal O(dN\min\{d, N\})$. For large $d$ and $N$, this procedure is computationally intensive and thus can be restrictive in practice. 
	To address this issue, we consider an approximate version of the best rank-$p$ approximation without taking the SVD of $\X$. 
	
	\begin{definition}
		A matrix $\widetilde{\X}_p$ is a $(1+\epsilon)$ rank-$p$ approximation to $\X$ if $\rank(\widetilde{\X}_p) \leq p$ and 
		\begin{align}
			\|\X-\widetilde{\X}_p\|_2\leq (1+\epsilon)\|\X-\X_p\|_2,\label{lra}
		\end{align}
		where $\X_p$ is the best rank-$p$ approximation to $\X$ as defined in \eqref{Xp}. 
	\end{definition}
	
	Before turning to discuss how to find such an $\widetilde{\X}_p$, we consider a few consequences assuming its existence. 
	Similar to the previous discussion, we can apply AA to $\widetilde{\X}_p$, which can be efficiently represented using the SVD. 
	As will be seen shortly, computing the SVD of $\widetilde{\X}_p$ is much cheaper than $\X$ when $p$ is small. 
	On the other hand, let $\widetilde{\X}_p = \widetilde{\U}_p\widetilde{\bm\Sigma}_p\widetilde{\V}^T_p$ be the SVD of $\widetilde{\X}_p$ and define 
	\begin{align}
		\widetilde{\X} = \widetilde{\bm\Sigma}_p\widetilde{\V}^T_p  
		\qquad \implies \qquad  
		\widetilde{\X}_p  = \widetilde{\U}_p\widetilde{\X}. \label{svd-rep-Xp}
	\end{align}
	The following theorem quantifies the approximation error if we use $\widetilde{\X}$ in place of $\X$ for archetypal analysis:
	
	\begin{theorem}\label{thm:sketch}
		Let $\epsilon>0$. 
		Suppose $\widetilde{\X}_p$ is a $(1+\epsilon)$ rank-$p$ approximation to $\X$, and $\widetilde{\X}$ is the representation of $\widetilde{\X}_p$ under the left singular vectors.   
		Let $(\widetilde{\A}, \widetilde{\B})$ be a solution to the AA applied to $\widetilde{\X}$:
		\begin{align}
			\min_{\A\in\R_{\cs}^{N\times k},\B\in\R_\cs^{k\times N}}\frac{1}{\sqrt{N}}\left\|\widetilde{\X}-\widetilde{\X}\A\B\right\|_F. \label{s1}
		\end{align} 
		Then, 
		\begin{align*}
			\frac{1}{\sqrt{N}}\|\X-\X\widetilde{\A}\widetilde{\B}\|_F\leq \opt(\X) + 4(1+\epsilon)\sigma_{p+1}(\X).
		\end{align*}
	\end{theorem}
	\begin{proof}
		Proceeding similarly as proof of Theorem \ref{thm:tSVD} with $\X_p, \X_{-p}$ replaced by $\widetilde{\X}_p$ and $\widetilde{\X}_{-p} = \X-\widetilde{\X}_p$, respectively, 
		\begin{align*}
			\frac{1}{\sqrt{N}}\|\X-\X\widetilde{\A}\widetilde{\B}\|_F&\leq\opt(\X) +4\|\widetilde{\X}_{-p}\|_2\stackrel{\eqref{lra}}{\leq}\opt(\X) +4(1+\epsilon)\sigma_{p+1}(\X).
		\end{align*}
	\end{proof}
	
	Theorem \ref{thm:sketch} implies that for $\X$ with small best rank-$p$ approximation error, using the SVD representation of $\widetilde{\X}_p$ will only result in a small impact on prediction accuracy. 
	The following algorithm, due to Musco and Musco \cite{musco2015randomized}, provides a way to compute $\widetilde{\X}_p$ (i.e., $\widetilde{\X}$) via randomized block Krylov methods. 
	The details of the algorithm are given in Algorithm \ref{alg:RK}:
	\medskip
	
	\begin{algorithm}[H]
		\KwIn{$\X\in\R^{d\times N}$: data matrix, $p$: approximation rank, $s$: power parameter}
		\KwOut{$\widetilde{\X}_p$}
		\begin{algorithmic}[1]
			\STATE {generate $p$ random initializations: $\S\in\R^{N\times p}$, $\S_{ij}\stackrel{\text{i.i.d}}{\sim}\mathcal N(0, 1)$}
			\STATE{construct the Krylov subspace: $\K = [\X\S, (\X\X^T)\X\S, \cdots, (\X\X^T)^{s-1}\X\S]\in\R^{d\times (sp)}$}
			\STATE{compute the QR decomposition for $\K$: $\K = \Q\bm R$}
			\STATE{compute the SVD of $\X_{\emd} = \X^T\Q$: $\X_{\emd} = \U_\emd\bm\Sigma_\emd\V^T_\emd$}
			\STATE{compute $\widetilde{\X}$: $\widetilde{\X}_p = \LL\LL^T\X$, with $\LL = \Q\V_\emd [:, 1:p]$}
		\end{algorithmic}
		\caption{Block Krylov Iteration \cite{musco2015randomized}} 
		\label{alg:RK}
	\end{algorithm}
	\medskip

	For moderately large $s$, with high probability, $\widetilde{\X}_p$ returned by Algorithm \ref{alg:RK} is a good approximation to $\X_p$: 
	
		\begin{lemma}[\cite{musco2015randomized}]\label{Krylov}
		For $\epsilon, \delta>0$, there exist absolute constants $C_1, C_2>0$ such that if 
		\begin{align}
			&s>\frac{C_1}{\sqrt{\epsilon}}\log \left(\frac{N}{\epsilon\delta}\right)&p\geq C_2\log\left(\frac{4}{\delta}\right),\label{mycond}
		\end{align} 
		then for with probability at least $1-\delta$, the $\widetilde{\X}_p$ in Algorithm \ref{alg:RK} satisfies \eqref{lra}. 
	\end{lemma}
	
	\begin{proof}
		Lemma \ref{Krylov} is a probabilistic version of \cite[Theorem 1]{musco2015randomized} where a fixed probability (0.99) is used instead of $1-\delta$ for an arbitrary $\delta$. Nevertheless, the proof is the similar except one needs to apply sharp concentration inequalities to bound extreme singular values of Gaussian matrices \cite[Corollary 7.3.3]{vershynin2018high}, \cite[Theorem 1.2]{rudelson2008littlewood} to control the failure probability. 
	\end{proof}
	
	\begin{remark}
		Other randomized low-rank approximation algorithms may also be used in place of Algorithm \ref{alg:RK}. 
		For example, one can use the randomized simultaneous iteration to compute $\widetilde{\X}_p$ \cite{halko2009finding, woodruff2014sketching}. 
		Under the same approximation error $\epsilon$ and failure probability $\delta$, the sample complexity of this method has a slightly worse dependence on $\epsilon$ (i.e. $s = \mathcal O(\log (N/\epsilon\delta)/\epsilon)$) than  \eqref{mycond}. 
		As such theoretical discrepancy was also manifested in several empirical studies in \cite{musco2015randomized}, we use Algorithm \ref{alg:RK} to compute $\widetilde{\X}_p$ in this article. 
	\end{remark}
	
	\begin{remark}
	The desired low-rank approximation $\widetilde{\X}_p$ is computed under the spectral norm, which is necessary in the derivation of approximation error in Theorem \ref{thm:sketch}. 
	Other randomized algorithms based on oblivious sketching \cite{sarlos2006improved, woodruff2014sketching, cohen2015dimensionality} or leverage score sampling \cite{cohen2017input} only produce low-rank approximations under the Frobenius norm. 
	Since an error bound under the Frobenius norm does not imply a similar bound under the spectral norm, these methods do not directly work for the problem considered in this paper. 
	\end{remark}  
	
	To apply Theorem \ref{thm:sketch}, we need to compute the SVD representation of the low-rank approximation matrix $\widetilde{\X}_p$, that is, $\widetilde{\X}$, rather than $\widetilde{\X}_p$ itself; see \eqref{svd-rep-Xp}.
	Since the output of Algorithm \ref{alg:RK} is the full low-rank matrix $\widetilde{\X}_p$, 
	finding its SVD representation may incur additional computational cost for our purpose. 
	However, in Algorithm \ref{alg:RK}, $\widetilde{\X}$ can be read off the shelf as $\widetilde{\X} = \bm\Sigma_\emd[1:p, 1:p](\U_\emd[:, 1:p])^T$, where $\bm\Sigma_\emd$ and $\U_\emd$ are computed in step 4. 
	Thus, the total cost for $\widetilde{\X}_p$ is the computational complexity for the first four steps in Algorithm \ref{alg:RK}.
	
	\begin{theorem}\label{time1}
		Let $\widetilde{\X}$ be the SVD representation of $\widetilde{\X_p}$ in Algorithm \ref{alg:RK}. Then, the computational complexity for $\widetilde{\X}$ is $\mathcal O(dNps + dp^2s^2+Nps\min\{N, ps\})$. 
	\end{theorem}
	
	\begin{proof}
		We only sketch the proof; more details can be found in \cite{musco2015randomized}.
		Step 1 in Algorithm \ref{alg:RK} generates a random Gaussian matrix which takes time $\mathcal O(Np)$.
		Step 2 computes the Krylov subspace basis which takes time $\mathcal O(dNps)$.
		The QR decomposition of $\K$ in step 3 takes time $\mathcal O(dp^2s^2)$. 
		In step 4, we first compute $\X_\emd$, which takes time $\mathcal O(dNps)$, then compute the SVD of $\X_\emd$, which takes time $\mathcal O(Nps\cdot\min\{N, ps\})$. 
	        Computing $\widetilde{\X} = \bm\Sigma_\emd[1:p, 1:p](\U_\emd[:, 1:p])^T$ takes time $\mathcal O(Np)$.
	\end{proof}
	\begin{remark}
		When $ps\ll \min\{d, N\}$, the computational complexity of $\widetilde{\X}$ becomes $\mathcal O(dNps)$, which is significantly smaller than $\mathcal O(dN\min\{d, N\})$.
	\end{remark}

	Setting $\epsilon = 1$ in Lemma \ref{Krylov} and combining Theorem \ref{thm:sketch}, we have the following result:
	
	\begin{theorem}\label{cor:1} 
		Let $\widetilde{\X}$ be the SVD representation of $\widetilde{\X}_p$ returned by Algorithm \ref{alg:RK}.
		Let $(\widetilde{\A}, \widetilde{\B})$ be a solution to the AA for $\widetilde{\X}$: 
		\begin{align}
			\min_{\A\in\R_{\cs}^{N\times k},\B\in\R_\cs^{k\times N}}\frac{1}{\sqrt{N}}\|\widetilde{\X}-\widetilde{\X}\A\B\|_F. 
		\end{align}
		For $\delta>0$, if $p$ satisfies \eqref{mycond} and $s>C\log(N/\delta)$ for some absolute constant $C>0$, then with probability at least $1-\delta$,   
		\begin{align*}
			\frac{1}{\sqrt{N}}\|\X-\X\widetilde{\A}\widetilde{\B}\|_F\leq \opt(\X) + 8\sigma_{p+1}(\X).
		\end{align*} 
	\end{theorem}
	\begin{remark}
		Fixing $\delta$ small, say $\delta = 0.01$, $s=\mathcal O(\log N)$. 
		The computational complexity of $\widetilde{\X}$ is $\mathcal O(dNp\log N + dp^2\log^2 N+Np\log N\min\{N, p\log N\})$. 
		Consequently, for data $\X$ that can be well approximated via low-rank matrices with approximation rank $p\ll N$, using Algorithm \ref{alg:RK} can effectively reduce the dimension of AA. 
	\end{remark}

	\section{Representation cardinality reduction}\label{sec:card}
	
	We now consider the situation where the dataset has a large cardinality, i.e., $N>d$.  
	In this case, to reduce computational complexity, we propose to use a parsimonious subset of points in $\X$ to approximately represent $\co(\X)$, i.e., we wish to find a small subset $T\subset [N]$ such that 
	\begin{align}
		\co(\X_T)\approx \co(\X),\label{app:conv}
	\end{align}
	where $\X_T:=\X[:, T]$ and $\approx$ will be made rigorous later. We will refer to $\co(\X_T)$ as an \emph{approximate convex hull} of $\X$. 
	
	The idea of using subsets of $\X$ (i.e. extreme points) to represent $\co(\X)$ has been considered in \cite{thurau2011convex,mair2017frame}, where exact equality in \eqref{app:conv} is expected.
	Here we only ask for approximate representation of $\co(\X)$ (allowing for a small approximation error), so that it is possible to further reduce the cardinality of the representation set for the archetypes (Figure \ref{fig:demo}). 
	
	\begin{figure}[htbp]
		\begin{center}
			\includegraphics[width=0.45\textwidth]{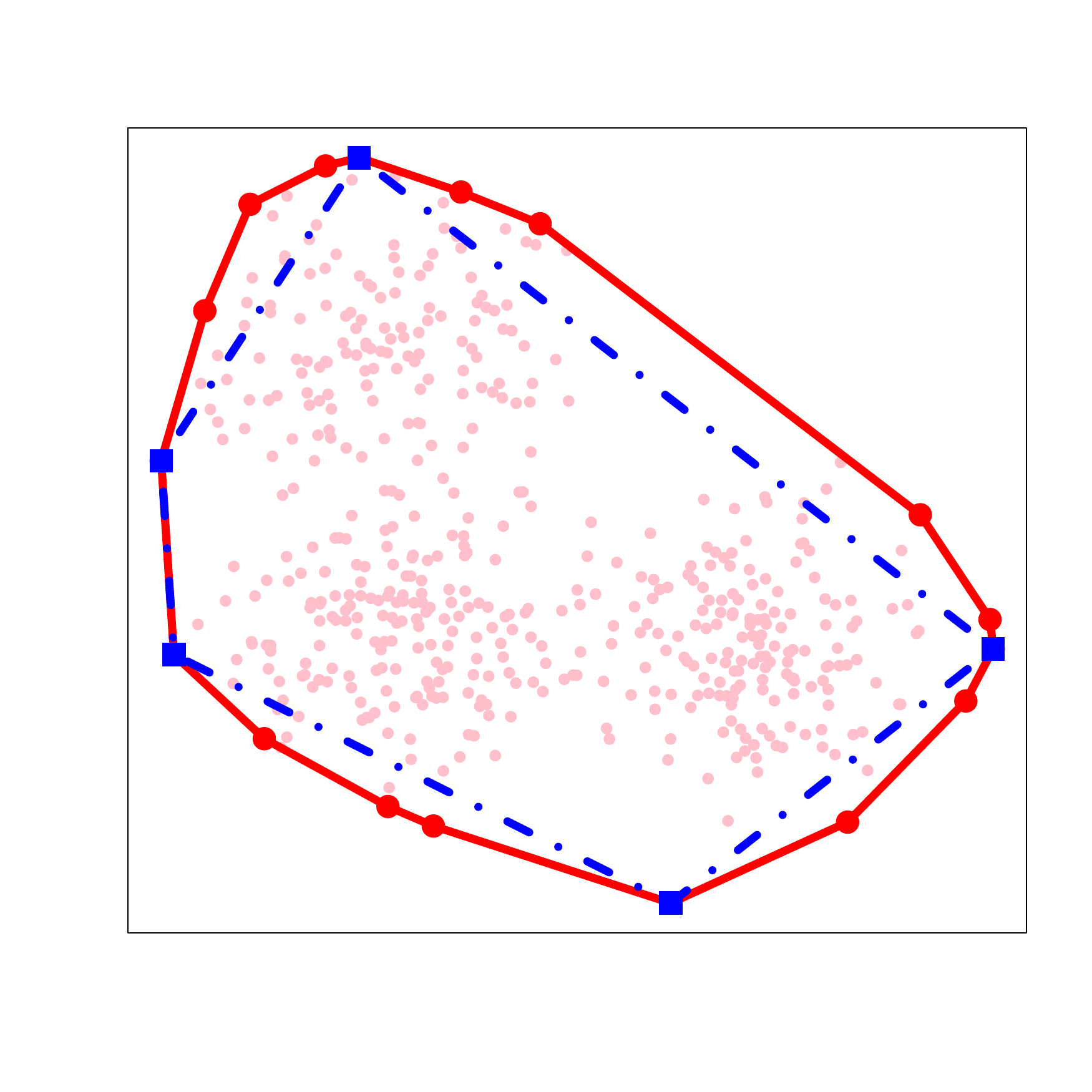}
			\caption{An example of the convex hull (red solid curves) and an approximate convex hull (blue dashed curve) of a randomly generated dataset. }
			\label{fig:demo}
		\end{center}
	\end{figure}
	
	Similar to the discussion in the previous section, we first give a few consequences assuming $\X_T$ exists. 
	
	\begin{definition}
		We say that $\X_T$ is an $\epsilon$-approximate convex hull of $\X$ if
		\begin{align}
			d_H(\co(\X_T), \co(\X))\leq \epsilon, \label{hau}
		\end{align}
		where $d_H (X,Y) := \max\left\{\,\sup_{x \in X} d(x,Y),\, \sup_{y \in Y} d(X,y) \,\right\}$ is the Hausdorff distance. 
	\end{definition}
	
	\begin{theorem}\label{thm:ach}
		For $\epsilon>0$ and $T\subseteq [N]$, suppose $\X_T$ is a $(\opt(\X)\cdot\epsilon)$-approximate convex hull of $\X$.
		Consider the following AA optimization problem constrained to $\co(\X_T)$: 
		\begin{align}
			\min_{\A\in\R_{\cs}^{|T|\times k},\B\in\R_\cs^{k\times N}}\frac{1}{\sqrt{N}}\|\X-\X_T\A\B\|_F. \label{3}
		\end{align}
		Then, the archetype points given by the solution of $\eqref{3}$ provide a $(1+\epsilon)$-approximation to the solution for \eqref{1} in terms of prediction errors:
		\begin{align*}
			\min_{\A\in\R_{\cs}^{|T|\times k},\B\in\R_\cs^{k\times N}}\frac{1}{\sqrt{N}}\|\X-\X_T\A\B\|_F\leq (1+\epsilon)\opt(\X).
		\end{align*}
	\end{theorem}  
	
	\begin{proof}
		For an optimal solution $(\A_\star, \B_\star)$ of \eqref{1} that resides on the boundary of $\co(\X)$ (such a solution always exists \cite{cutler1994archetypal}), consider the projection of each column of $\X\A_\star$ to $\co(\X_T)$, and denote the projected points as $\bm Z$. Note that $\bm Z$ is well-defined as $\co(\X_T)\subset\co(\X)$. 
		By the triangle inequality, the distance between each column of $\X$ and $\co(\bm Z)$ is bounded by the sum of the distance between the column of $\X$ and $\co(\X\A_\star)$ and $d_H(\X\A_\star, \bm Z)$. 
		Since $\X_T$ gives an $(\opt(\X)\cdot\epsilon)$-approximate convex hull of $\X$, 
		\begin{align}
			d_H(\X\A_\star, \bm Z)\leq d_H(\co(\X), \co(\X_T))\leq\epsilon\cdot\opt(\X) = \frac{\epsilon}{\sqrt{N}}\|\X-\X\A_\star\B_\star\|_F.\label{lll}
		\end{align}
		It follows from direct computation that
		\begin{align*}
			\min_{\A\in\R_{\cs}^{|T|\times k},\B\in\R_\cs^{k\times N}}\frac{1}{N}\|\X-\X_T\A\B\|_F^2&\leq \min_{\B\in\R_\cs^{k\times N}}\frac{1}{N}\|\X-\bm Z\B\|_F^2\\
			&= \frac{1}{N}\sum_{i\in [N]}d(x_i, \co(\bm Z))^2\\
			&\leq\frac{1}{N}\sum_{i\in [N]}(d(x_i, \co(\X\A_\star)) + d_H(\X\A_\star, \bm Z))^2\\
			&\leq(1+\epsilon)^2\cdot\frac{1}{N}\|\X-\X\A_\star\B_\star\|_F^2,
		\end{align*}
		where the last inequality follows from \eqref{lll} and Cauchy-Schwarz inequality. Taking the square root on both sides completes the proof. 
	\end{proof}
	
	Theorem \ref{thm:ach} establishes an approximate equivalence between the solutions of \eqref{3} and \eqref{1} in terms of objective values. 
	Compared to \eqref{1}, the dimension of $\A$ is significantly reduced provided $|T|\ll N$, while the dimension of $\B$ stays unchanged. 
	
	To see the computational gain from solving \eqref{3} instead of \eqref{1}, recall the alternating minimization in Section \ref{sec:2}.
	When $\B$ is fixed and $\A$ is updated, one needs to solve $k$ quadratic programming problems with variable dimensions equal to the number of rows of $\A$. 
	For certain optimization methods such as the ellipsoid method, the complexity of quadratic programming problems with positive-definite quadratic matrix has weakly polynomial time (of the variable dimension) \cite{kozlov1979polynomial}. 
	Thus, when $|T|\ll N$, a notable acceleration is expected for the subroutine of updating $\A$, which justifies the significance of using a parsimonious subset of points to represent the archetypes.     
	
	On the other hand, when $\A$ is fixed and $\B$ is updated, one needs to compute the projection of each column of $\X$ to $\co(\X\A)$. This step is the same in both \eqref{1} and \eqref{3} and consists of $N$ independent quadratic programming problems with variable dimension $k$. 
	In this case, it is possible to take an additional step of acceleration via parallelization combined with the coreset approximation \cite{mair2019coresets}, which reduces the computation of $N$ projection coefficient vectors to a small subset of points in $\X$ with appropriate weights, similar to the ideas of quadrature.
	Indeed, given a coreset $\X_C\subset\X$ and appropriate weight diagonal matrix $\bm W$, one can approximate the objective function in \eqref{1} with $\|\bm W\X_C-\bm W\X\A\B\|_F/\sqrt{N}$. Note the complexity of the subproblem for updating $\B$ in the alternating minimization algorithm is proportional to the number of points in the objective function. Therefore, when $|\X_C|\ll |\X|$, the step of solving the $\B$-subproblem can be significantly accelerated. 
	Combining the idea of coreset with \eqref{lll} yields an approximate objective function $\|\bm W\X_C-\bm W\X_T\A\B\|_F/\sqrt{N}$, which has significantly reduced complexity when solved by the alternating minimization algorithm. The details are not discussed here. 
	
	We next discuss how to find a ``small'' subset $T$ such that \eqref{hau} is satisfied. 
	Note that to represent $\co(\X)$, it suffices to consider the extreme points of $\X$. 
	In other words, we will find a subset $\X_T\subset \ex(\X)$ whose convex hull can well approximate $\co(\X)$. 
	As will be seen below, this procedure can be effectively implemented by taking random projections. 
	Indeed, random projections are linear maps whose inverse image of the extreme points of a convex set are a subset of the extreme points of the inverse image of that convex set \cite{lax2002functional}.
	Similar ideas have been used in the empirical study of archetypal analysis to seek extreme points \cite{thurau2011convex, damle2017geometric}. 
	
	Finding all the extreme points of $\co(\X)$ may itself be computationally demanding unless $\ex(\X)$ is small. 
	When $\co(\X)$ can be well approximately using a few extreme points, it is desired to single them out to further shrink the complexity of the problem at a small sacrifice of accuracy.     
	To this end, we need to know which extreme points are more important than the others in terms of composing $\co(\X)$.
	The following result, which originally appeared in \cite{graham2017approximate}, is precisely what is needed here.  
	
	Observe that under a random projection $v\in\mathbb S^{d-1}$, the points in $\X$ have projected values $\{\langle x_i, v\rangle\}_{i\in [N]}$, which with probability one have a unique maximum. 
	The inverse image of the maximum is an element in $\ex(\X)$.  
	Thus, throwing away a null set, we can partition the unit sphere $\mathbb S^{d-1}$ as follows: 
	\begin{align*}
		&\mathbb S^{d-1} = \bigsqcup_{x\in\ex(\X)}V_x& V_x = \{v\in\mathbb S^{d-1} : v^Tx\geq v^Tx_i, i\in [N]\}.
	\end{align*} 
	For $x\in\ex(\X)$, its curvature is defined as $$\kappa(x) = \frac{|V_x|}{|\mathbb S^{d-1}|},$$ 
	which is the relative area of the directions that distinguish $x$ as the maximum to the unit sphere in $\R^d$. 
	By definition, points with larger curvature are more likely to be sampled if $v$ is uniformly drawn from $\mathbb S^{d-1}$; 
	in fact, they are also more `important' as specified by the following lemma \cite[Theorem 3.4]{graham2017approximate}: 
	
	\begin{lemma}\label{th:high}
		Let $S\subset [N]$. 
		Suppose that both $\co(\X_S)$ and $\co(\X)$ are non-degenerate (i.e., with nonempty interior), and $R: =\max_{i\in [N]}\|x_i\|_2$. 
		Then, 
		\begin{align}
			&d_H(\co(\X), \co(\X_S))\leq \min\left\{\sqrt{2}\pi (2\omega)^{\frac{1}{d-1}}, 2\right\}\cdot R & \omega = \sum_{i\in [N]\setminus S}\kappa(x_i). \label{alex}
		\end{align}
	\end{lemma}
	
	As a result, to compute a sparse approximate convex hull, it suffices to use high-curvature points to approximately represent $\co(\X)$.  
	To find high-curvature points, we apply a Monte-Carlo (MC) procedure to estimate the curvature of each point and then truncate at some thresholding parameter. 
	The details are given in Algorithm \ref{alg:RPACH}:
	\medskip
	
	\begin{algorithm}[H]
		\KwIn{$\{x_i\}_{i\in [N]}$: dataset, $M$: number of projections, $\eta$: approximation accuracy}
		\KwOut{approximate convex hull $\co(\{x_i\}_{i\in T})$}
		\begin{algorithmic}[1]
			\STATE {$e_j = 0$ for $j\in [N]$}.
			\FOR{$i = 1, \cdots, M$}{
				\STATE {$v_i\sim\text{Uniform}(\mathbb S^{d-1})$}
				\STATE {$u_i = \arg\max_{j} v^T_ix_j$}
				\STATE {$e_{u_i}\gets e_{u_i} + 1$}
			}
			\ENDFOR
			\STATE sort $\{e_j\}_{j\in [N]}$ in decreasing order as $e^{(1)}\geq\cdots\geq e^{(N)}$ 
			\STATE compute $L =\min\left\{\ell: \frac{1}{M}\sum_{j\in [\ell]}e^{(j)}>1-\eta/3\right\}$
			\STATE compute $L\gets \max\left\{L, d+1\right\}$ 
			\STATE let $T$ be the index set of $e^{(1)}, \cdots, e^{(L)}$ and return $\{x_j\}_{j\in T}$
		\end{algorithmic}
		\caption{Approximate Convex Hull} 
		\label{alg:RPACH}
	\end{algorithm}
	\medskip

	 A similar MC method based on a different truncation rule has been proposed \cite[Algorithm 1]{graham2017approximate}, where points are removed whenever their estimated curvatures are below some fixed threshold. 
	To ensure that the remaining points have large cumulative curvature, this algorithm requires the thresholding parameter to be overly small, leaving most points unremoved.
	To facilitate parsimony, Algorithm \ref{alg:RPACH} first sorts points based on their estimated curvatures, then truncates based on the estimated cumulative curvatures.

	The computational complexity of Algorithm \ref{alg:RPACH} can be easily obtained from direct computation:
	
	\begin{theorem}\label{time2}
		The computational complexity for Algorithm \ref{alg:RPACH} is $\mathcal O(MNd+ N\log N)$. 
	\end{theorem} 
	\begin{proof}
		The MC procedure in Algorithm \ref{alg:RPACH} (step 2 to 6) involves $M$ repetitions of computing $N$ $d$-dimensional vector inner product and finding the (index of the) maximum of the projected points, which takes time $\mathcal O(MNd)$ in total.
		Step 7 is a simple sorting that has complexity $\mathcal O(N\log N)$ using Merge sort \cite{knuth1997art}. 
	\end{proof}
	
	We will show that for large $M$, with high probability, the output of Algorithm \ref{alg:RPACH} satisfies \eqref{hau} with $\epsilon = \min\left\{\sqrt{2}\pi \eta^{\frac{1}{d-1}}, 2\right\}\cdot R$.  
	Without loss of generality, in the following discussion we assume $|\ex(\X)| = h$ and 
	\begin{align}
		\kappa(x_1)\geq\kappa(x_2)\geq\cdots\geq \kappa(x_h)>\kappa (x_{h+1})=\cdots =\kappa(x_N) = 0.\label{assume}
	\end{align} 
	We have the following theorem:
	
	\begin{theorem}\label{main}
		Let $T$ be the subset returned by Algorithm \ref{alg:RPACH}, and $R = \max_{i\in [N]}\|x_i\|_2$. 
		Suppose $\co(\X_D)$ is non-degenerate for every $D\subset [N]$ with $|D|>d$. 
		Denote $q$ as the smallest integer such that $\sum_{i\in [q]}\kappa(x_i)\geq 1-\eta/18$:
		\begin{align}
			q: = \min\left\{j: \sum_{i\in [q]}\kappa(x_i)\geq 1-\frac{\eta}{18}\right\},\label{back}
		\end{align} 
		and the truncation gap 
		\begin{align*}
			\Delta: = \kappa(x_{q})-\kappa(x_{q+1})>0.
		\end{align*}
		If
		\begin{align}
			M\geq \max\left\{\frac{324q^2}{\eta^2}, \frac{4}{\Delta^2}\right\}\log\left(\frac{3N}{\sqrt{\delta}}\right),\label{Mbdd}
		\end{align}
		then with probability at least $1-\delta$, $|T|\leq \max\{q, d+1\}$ and 
		\begin{align}
			d_H(\co(\X_T), \co(\X))\leq \min\left\{\sqrt{2}\pi \eta^{\frac{1}{d-1}}, 2\right\}\cdot R.\label{des1}
		\end{align}
	\end{theorem}
	
	\begin{remark}
		Setting the upper bound in \eqref{des1} equal to $\opt(\X)\epsilon$ yields 
		\begin{align*}
			M\geq \max\left\{324q^2\left(\frac{2\pi^2R^2}{\opt(\X)^2\epsilon^2}\right)^{d-1}, \frac{4}{\Delta^2}\right\}\log\left(\frac{3N}{\sqrt{\delta}}\right),
		\end{align*}
		which has an unpleasant but expected exponential dependence on $d$ (curse of dimensionality).  
		For datasets with low-dimensional structure, i.e., well approximated via rank-$p$ matrices with $p\ll d$, it is possible to use ideas in Section \ref{sec:data} to improve the exponential dimension dependence to $p$ (Algorithm \ref{alg:AAA}).  
	\end{remark}
	
	
	\begin{proof}[Proof of Theorem \ref{main}]
		Note that step 9 in Algorithm \ref{alg:RPACH} ensures that $\co(\X_T)$ is non-degenerate. 
		Therefore, to show \eqref{des1}, by \eqref{alex}, it suffices to show $\sum_{i\in [N]\setminus T}\kappa(x_i)\leq \eta/2$, or equivalently, $\sum_{i\in T}\kappa(x_i)\geq 1-\eta/2$. 
		
		We first show that for $M$ satisfying \eqref{Mbdd}, with high probability, the estimated curvatures $e_j/M$ are close to their expectations for all reasonably large $e_j$. 
		Note for every $j\in [q]$, $e_j$ is a sum of $M$ independent Bernoulli random variables with parameter $\kappa(x_j)$, and the tail sum $\sum_{j>q}e_j$ is a sum of $M$ independent Bernoulli random variables with parameter $\sum_{j>q}\kappa(x_j)<\eta/18$.
		Thus, by Hoeffding's inequality \cite{hoeffding1994probability}, 
		\begin{align*}
			\P\left[\left|\frac{1}{M}e_j-\kappa(x_j)\right|\leq \frac{\eta}{18q}\right]&\geq 1-2\exp\left(-\frac{M\eta^2}{162q^2}\right)\\
			\P\left[\left|\frac{1}{M}\sum_{j>q}e_j-\sum_{j>q}\kappa(x_j)\right|\leq \frac{\eta}{18q}\right]&\geq 1-2\exp\left(-\frac{M\eta^2}{162q^2}\right).
		\end{align*}
		Taking a union bound over $j\in [q]$ and combining the two inequalities yields
		\begin{align}
			&\P\left[\left\{\max_{j\in [q]}\left|\frac{1}{M}e_j-\kappa(x_j)\right|, \left|\frac{1}{M}\sum_{j>q}e_j-\sum_{j>q}\kappa(x_j)\right|\right\}\leq\frac{\eta}{18q}\right]\nonumber\\
			\geq&\ 1-2(q+1)\exp\left(-\frac{M\eta^2}{162q^2}\right)\geq 1-4q\exp\left(-\frac{M\eta^2}{162q^2}\right).\label{goodevent} 
		\end{align}
		The right-hand side in \eqref{goodevent} can be further lower bounded by $1-\delta/2$ if $M$ satisfies \eqref{Mbdd}. 
		
		We next show that for large $M$, with high probability, the largest $q$ terms of $e_j$, i.e., $e^{(1)}, \cdots, e^{(q)}$, coincide with $\{x_j\}_{j\in [q]}$. 
		Particularly, denoting the index of $e^{(j)}$ as $\ell_j$, we will show $[q] = \{\ell_1, \cdots, \ell_q\}$. 
		Note that $[q] = \{\ell_1, \cdots, \ell_q\}$ if and only if the following probabilistic event occurs:
		\begin{align*}
			C_q: = \left\{\min_{i\leq q}e_i>\max_{j>q}e_j\right\}.
		\end{align*}
		Since for every $i\leq q$ and $j>q$, $e_i-e_j$ is a sum of $M$ i.i.d. random variables $Z$, where $Z=1$ with probability $\kappa(x_i)$, $Z=-1$ with probability $\kappa(x_j)$, and $Z=0$ otherwise.
		Thus, we can bound the probability of $C_q$ from below with another application of Hoeffding's inequality:
		\begin{align*}
			\P\left[C_q\right] = 1-\P\left[C^\complement_q\right]&\geq 1-\sum_{i\leq q, j>q}\P\left[e_i-e_j<0\right]\\
			&\geq 1-\sum_{i\leq q, j>q}\P\left[e_i-e_j-\E[e_i-e_j]<-M\Delta\right]\\
			&\geq 1-\frac{h^2}{4}\exp\left(-\frac{M\Delta^2}{2}\right)\\
			&\geq 1-\frac{N^2}{4}\exp\left(-\frac{M\Delta^2}{2}\right),
		\end{align*}
		which is lower bounded by $\delta/2$ if $M$ satisfies \eqref{Mbdd}. 
		Taking a union bound, for $M$ satisfying \eqref{Mbdd}, both the event in \eqref{goodevent} and $C_q$ occur with probability at least $1-\delta$. 
		
		To finish the proof, it suffices to show that conditional on both events, (i) $\sum_{i\in T}\kappa(x_i)\geq 1-\eta/2$ and (ii) $|T|\leq \max\{q, d+1\}$. 
		Let $T_- = T\setminus\{\ell_L\}$.
		Conditional on the event in \eqref{goodevent}, it follows from the stopping rule in Algorithm \ref{alg:RPACH} that
		\begin{align}
			\sum_{j\in T}\kappa(x_j)&\stackrel{\eqref{goodevent}}{\geq} \sum_{j\in T\cap [q]}\left(\frac{1}{M}e_j-\frac{\eta}{18q}\right) + \sum_{j\in T\cap [N]\setminus [q]}\kappa(x_j)\nonumber\\
			&\geq \sum_{j\in T\cap [q]}\left(\frac{1}{M}e_j-\frac{\eta}{18q}\right) +  \sum_{j\in [N]\setminus [q]}\kappa(x_j)-\sum_{j\in [N]\setminus [q]}\kappa(x_j)\nonumber\\
			&\stackrel{\eqref{goodevent}, \eqref{back}}{\geq} \sum_{j\in T\cap [q]}\left(\frac{1}{M}e_j-\frac{\eta}{18q}\right) +  \sum_{j\in [N]\setminus [q]}\frac{1}{M}e_j-\frac{\eta}{18q}-\frac{\eta}{18}\nonumber\\
			&\geq \sum_{j\in T\cap [q]}\left(\frac{1}{M}e_j-\frac{\eta}{18q}\right) +  \sum_{j\in T\cap [N]\setminus [q]}\frac{1}{M}e_j-\frac{\eta}{18q}-\frac{\eta}{18}\nonumber\\
			&\geq 1-\frac{\eta}{3} - \frac{\eta}{18} -\frac{\eta}{18q}-\frac{\eta}{18}>1-\frac{\eta}{2}\nonumber,
		\end{align}
		which shows that (i) holds true. 
		
		To show (ii), it suffices to consider the case where $\min\left\{\ell: \frac{1}{M}\sum_{j\in [\ell]}e^{(j)}>1-\eta/3\right\}>d+1$, since otherwise $L = d+1\leq\max\{q, d+1\}$.  
		In this case, the stopping rule in Algorithm \ref{alg:RPACH} tells us 
		\begin{align}
			\sum_{j\in T_-}\kappa(x_j)&\stackrel{\eqref{goodevent}}{\leq} \sum_{j\in T_-\cap [q]}\left(\frac{1}{M}e_j+\frac{\eta}{18q}\right) + \sum_{j\in T_-\cap [N]\setminus [q]}\kappa(x_j)\nonumber\\
			&\stackrel{\eqref{back}}{\leq} \sum_{j\in T_-\cap [q]}\left(\frac{1}{M}e_j+\frac{\eta}{18q}\right) +\frac{\eta}{18}\nonumber\\
			&\leq\sum_{j\in T_-}\frac{1}{M}e_j+\frac{\eta}{18} +\frac{\eta}{18}\stackrel{\eqref{back}}{<}\sum_{j\in [q]}\kappa(x_j).\label{1111}
		\end{align} 
		Further conditioning on $C_q$, we have $T_-\subset [q]$ or $ [q]\subset T_-$. But the latter cannot happen owing to \eqref{1111}. This implies $|T| =  |T_-|+1\leq q = \max\{q, d+1\}$, establishing (ii). 
	\end{proof}
	
	\section{An approximate AA algorithm}\label{sec:comb}
	
	Putting results in Section \ref{sec:data} and \ref{sec:card} together, we have the following approximate algorithm for archetypal analysis (AAA): 
	\medskip
	
	\begin{algorithm}[H]
		\KwIn{$\{x_i\}_{i\in [N]}$: dataset, $k$: number of archetypes, $p$: approximation rank, $s$: Krylov subspace parameter, $M$: number of projections, $\eta$: approximation accuracy}
		\KwOut{ an approximate solution to \eqref{1}}
		\begin{algorithmic}[1]
			\STATE {generate $p$ random initializations: $\S\in\R^{N\times p}$, $\S_{ij}\stackrel{\text{i.i.d}}{\sim}\mathcal N(0, 1)$}
			\STATE{construct the Krylov subspace: $\K = [\X\S, (\X\X^T)\X\S, \cdots, (\X\X^T)^{s-1}\X\S]\in\R^{d\times (sp)}$}
			\STATE{compute the QR decomposition for $\K$: $\K = \Q\bm R$}
			\STATE{compute the SVD of $\X_{\emd} = \X^T\Q$: $\X_{\emd} = \U_\emd\bm\Sigma_\emd\V^T_\emd$}
			\STATE{form approximate reduced SVD representation: $\widetilde{\X} = \bm\Sigma_\emd[1:p, 1:p](\U_\emd[:, 1:p])^T$}
			\STATE{apply Algorithm \ref{alg:RPACH} to $\widetilde{\X}$ with parameters ($M, \eta$) to find a subset of $T\subset [N]$}
			\STATE{solve the reduced archetypal analysis problem: 
				\begin{align*}
					(\widetilde{\A}_\star, \widetilde{\B}_\star)\in\arg\min_{\widetilde{\A}\in\R_\cs^{|T|\times k}, \widetilde{\B}\in\R^{k\times N}_\cs}\frac{1}{\sqrt{N}}\|\widetilde{\X}-\widetilde{\X}_T\widetilde{\A}\widetilde{\B}\|_F
			\end{align*}}
			\STATE{extend $\widetilde{\A}_\star$ to an $\R^{N\times k}$ matrix by first creating a zero matrix $\A_{null}\in\R^{N\times k}$, then $\A_{null}[T, :] \gets \widetilde{\A}_\star$, and finally $\widetilde{\A}_\star\gets\A_{null}$}
			\STATE{return $\widetilde{\A}_\star$, $\widetilde{\B}_\star$}
		\end{algorithmic}
		\caption{Approximate Archetypal Analysis (AAA)} 
		\label{alg:AAA}
	\end{algorithm}
	\medskip
	
	Under appropriate assumptions on the input parameters, we have the following guarantee for the solutions computed by Algorithm \ref{alg:AAA}:
	
	\begin{theorem}\label{thm:last}
		Under the same assumptions in Theorem \ref{main} and $p\gtrsim\log(1/\delta)$, 
		if
		\begin{align}
			s &\geq C\log\left(\frac{N}{\delta}\right)&\eta = \left(\frac{\opt(\X)\epsilon}{\sqrt{2}\pi\max_{i\in [N]}\|x_i\|_2}\right)^{p-1}\label{mb}\\
			M&\geq \max\left\{\frac{324q^2}{\eta^2}, \frac{4}{\Delta^2}\right\}\log\left(\frac{3N}{\sqrt{\delta}}\right), 
		\end{align}
		where $C$ is the same constant as in Theorem \ref{cor:1}, $\opt(\X)$ is the optimum value of \eqref{1}, $q$, $\Delta$ are the same as defined in Theorem \ref{main}, then with probability at least $1-2\delta$, $|T|\leq \max\{q, p+1\}$, and the approximate archetypes $\X\widetilde{\A}_\star$ as well as the coefficient matrix $\widetilde{\B}_\star$ returned by Algorithm \ref{alg:AAA} satisfy
		\begin{align*}
			\frac{1}{\sqrt{N}}\|\X-\X\widetilde{\A}_\star\widetilde{\B}_\star\|_F\leq (1+\epsilon)\left(\opt(\X)+8\sigma_{p+1}(\X)\right). 
		\end{align*}
	\end{theorem}
	
	\begin{remark}
		According to Theorem \ref{time1} and Theorem \ref{time2}, the computational complexity of data dimensionality reduction (step 1 to step 5) and representation cardinality reduction (step 6) is $\mathcal O(dNp\log N +dp^2\log^2 N + Np\log N\min\{N, p\log N\})$ and $\mathcal O(MNp+N\log N) = \mathcal O\left((q^2\epsilon^{-2(p-1)}+\Delta^{-2} + N)\log N\right)$, respectively. 
		With probability at least $1-2\delta$, step 6 solves the reduced problem which has data dimension $p$ and representation cardinality $|T|\leq \max\{p+1, q\}$. 
		Thus, the overall complexity for Algorithm \ref{alg:AAA} is small if both $p$ and $q$ are small and $\Delta$ is away from $0$. 
		This corresponds to the scenario where $\X$ is approximately low-rank and has most of the curvature concentrated on a small subset.   
	\end{remark}
	
	\begin{proof}[Proof of Theorem \ref{thm:last}]
		Let $(\A_\star, \B_\star)$ and $(\widetilde{\A}, \widetilde{\B})$ be solutions to \eqref{1} and 
		\begin{align}
			\min_{\A\in\R_{\cs}^{N\times k},\B\in\R_\cs^{k\times N}}\frac{1}{\sqrt{N}}\|\widetilde{\X}-\widetilde{\X}\A\B\|_F,
		\end{align}
		respectively. 
		Under the assumptions on $\eta$ and $M$, Theorem \ref{thm:ach} and Theorem \ref{main} together imply that with probability at least $1-\delta$, $|T|\leq\max\{p+1, q\}$ and 
		\begin{align}
			\|\widetilde{\X}-\widetilde{\X}\widetilde{\A}_\star\widetilde{\B}_\star\|_F\leq (1+\epsilon)\|\widetilde{\X}-\widetilde{\X}\widetilde{\A}\widetilde{\B}\|_F\leq (1+\epsilon)\|\widetilde{\X}-\widetilde{\X}\A_\star\B_\star\|_F. \label{f1}
		\end{align}
		Let $\widetilde{\X}_p = \widetilde{\U}_p\widetilde{\X}$ and $\widetilde{\X}_p=\X-\widetilde{\X}_p$, where $\widetilde{\U}_p$ is the left singular vector matrix of the low-rank approximation given by Algorithm \ref{alg:RK}. 
		For $s$ satisfying \eqref{mb}, it follows from Lemma \ref{Krylov} that with probability at least $1-\delta$, 
		\begin{align}
			\|\widetilde{\X}_{-p}\|_2 = \|\X-\widetilde{\X}_p\|_2\leq 2\|\X-\X_p\|_2 = 2\sigma_{p+1}(\X),\label{f2}
		\end{align}
		where $\X_p$ is the best rank-$p$ approximation for $\X$. 
		Thus, both \eqref{f1} and \eqref{f2} hold with probability $1-2\delta$. 
		Conditioning on \eqref{f1} and \eqref{f2}, 
		the rest of the proof is similar to the computation in \eqref{etc}:
		\begin{align*}
			\|\X-\X\widetilde{\A}_\star\widetilde{\B}_\star\|_F &\leq  \|\widetilde{\X}_p-\widetilde{\X}_p\widetilde{\A}_\star\widetilde{\B}_\star\|_F + \|\widetilde{\X}_{-p}-\widetilde{\X}_{-p}\widetilde{\A}_\star\widetilde{\B}_\star\|_F\nonumber\\
			& =  \|\widetilde{\X}-\widetilde{\X}\widetilde{\A}_\star\widetilde{\B}_\star\|_F + \|\widetilde{\X}_{-p}-\widetilde{\X}_{-p}\widetilde{\A}_\star\widetilde{\B}_\star\|_F\nonumber\\
			&\stackrel{\eqref{f1}}{\leq}  (1+\epsilon)\|\widetilde{\X}-\widetilde{\X}\A_\star\B_\star\|_F + \|\widetilde{\X}_{-p}-\widetilde{\X}_{-p}\widetilde{\A}_\star\widetilde{\B}_\star\|_F\nonumber\\
			&= (1+\epsilon)\|\widetilde{\X}_p-\widetilde{\X}_p\A_\star\B_\star\|_F + \|\widetilde{\X}_{-p}-\widetilde{\X}_{-p}\widetilde{\A}_\star\widetilde{\B}_\star\|_F\nonumber\\
			&\leq (1+\epsilon)\|\X-\X\A_\star\B_\star\|_F + (1+\epsilon)\|\widetilde{\X}_{-p}-\widetilde{\X}_{-p}\A_\star\B_\star\|_F + \|\widetilde{\X}_{-p}-\widetilde{\X}_{-p}\widetilde{\A}_\star\widetilde{\B}_\star\|_F\\
			&\stackrel{\eqref{>>>}}{\leq}(1+\epsilon)\left(\|\X-\X\A_\star\B_\star\|_F + 2\|\widetilde{\X}_{-p}\|_F + 2\sqrt{N}\|\widetilde{\X}_{-p}\|_2\right)\\
			&\leq(1+\epsilon)\left(\|\X-\X\A_\star\B_\star\|_F + 4\sqrt{N}\|\widetilde{\X}_{-p}\|_2\right)\\
			&\stackrel{\eqref{f2}}{\leq}(1+\epsilon)\left(\|\X-\X\A_\star\B_\star\|_F + 8\sigma_{p+1}(\X)\sqrt{N}\right).
		\end{align*}
		Dividing both sides by $\sqrt{N}$ yields the desired result. 
	\end{proof}

	\section{Numerical experiments}\label{sec:num}
	
	In this section, we apply the proposed algorithm (Algorithm \ref{alg:AAA}) to compute the archetypes for three real datasets, including a time series dataset and two image datasets. 
	When implementing the alternating minimization algorithm for solving AA, we use the $k$-means to find an initial guess for the archetypes; the subproblems are solved using the existing package `\texttt{quadprog}' \cite{quadprog} in \texttt{R} \cite{R}. 
	The algorithm stops if the relative objective decrease falls below \texttt{1e-3}.  
	We will compare the computation time and accuracy of the following algorithms:
	\begin{itemize}
		\item (SVD-AA): Alternating minimization applied to the reduced singular value representation of $\X$ as in \eqref{svd-aa}, where truncation keeps $99.99\%$ of the variance of the data. SVD is implemented using the built-in function `\texttt{svd}' in \texttt{R}.   
		\item (AAA): Approximate archetypal analysis (Algorithm \ref{alg:AAA}), with $s = \lceil\log N\rceil$. 
		\item (archetypes): A function for archetypal analysis in the package \texttt{archetypes} in \texttt{R} \cite{aa1,aa2}, whose implementation is different from Algorithm \ref{alg:AM}. 
	\end{itemize}
	To ensure comparability of the results, we do not include other accelerated algorithms such as the active-subset solver \cite{chen2014fast} and the coreset approximation \cite{mair2019coresets}, 
	which have a different focus than our methods.  
	All reported results in this section were obtained on a Macbook Air with an M1 processor and 8GB of RAM. 
	
	%
	%
	
	\subsection{S\&P 500 cumulative log-returns}
	The Standard and Poor's 500 (S\&P 500) is a stock market index consisting of 500 large companies listed on stock exchanges in the United States. 
	It is one of the most commonly used equity indices to evaluate the financial market as well as the economy. 
	The companies that are selected for the S\&P 500 index are changing with time. 
	In this example, we consider a dataset comprised of $385$ companies that are currently in the S\&P 500 index by January 2022, with close price recorded from December 2011 to December 2021. 
	We compute for each column in $\X$ a $2515$-dimensional time series representing the cumulative log-return (CLR) of a company over ten years. 
	The CLR is calculated on a daily basis using the adjusted prices of stocks.
	Visualization of the dataset is given in the first plot in Figure \ref{fig:1}. 
	In the rest of the section, we assume that the CLR of each company in the S\&P 500 index can be decomposed with respect to a few distinct growth patterns that can be identified via AA. 
		
	\begin{figure}[htbp]
		\begin{center}
			\includegraphics[width=0.40\textwidth]{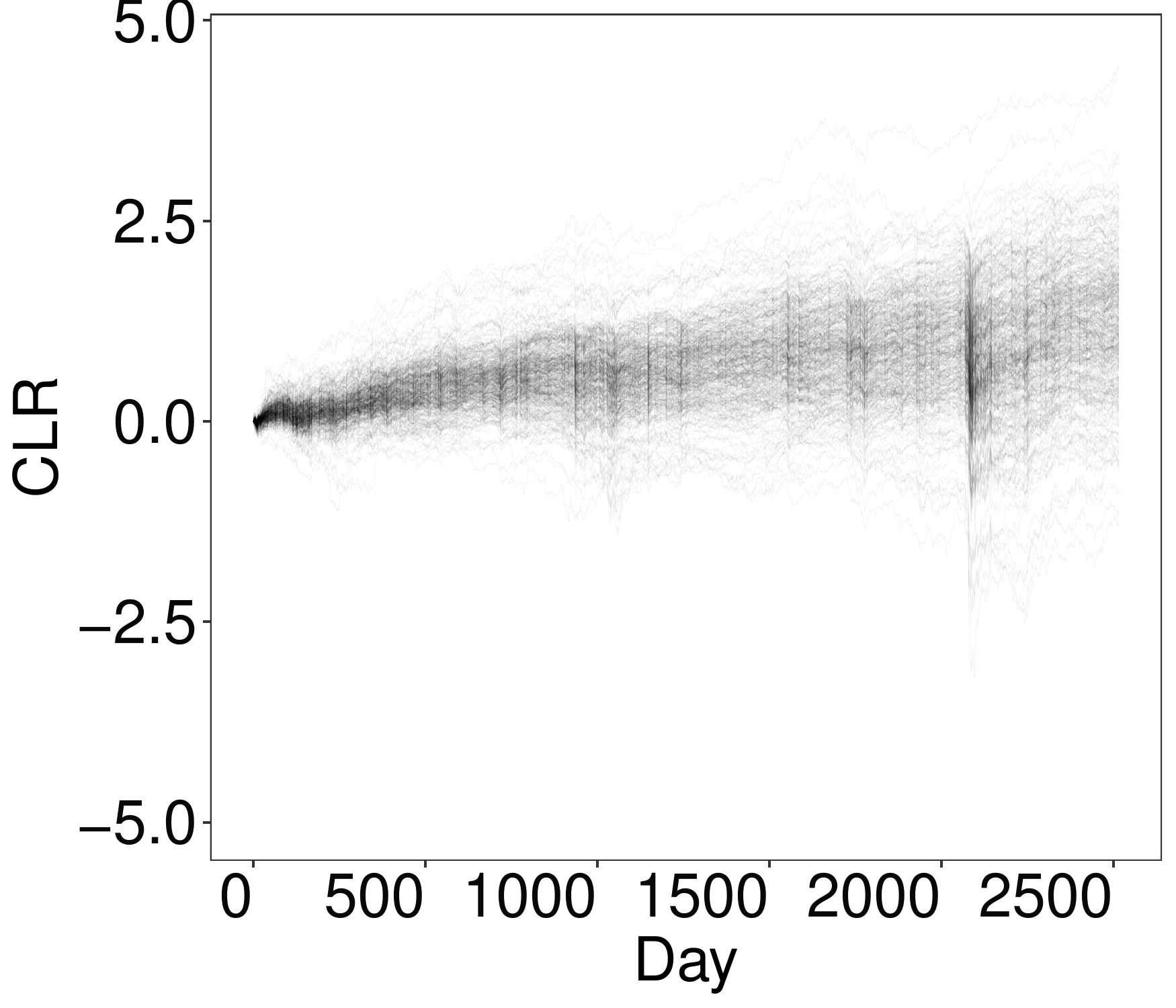}\hspace{2 cm}
			\includegraphics[width=0.40\textwidth]{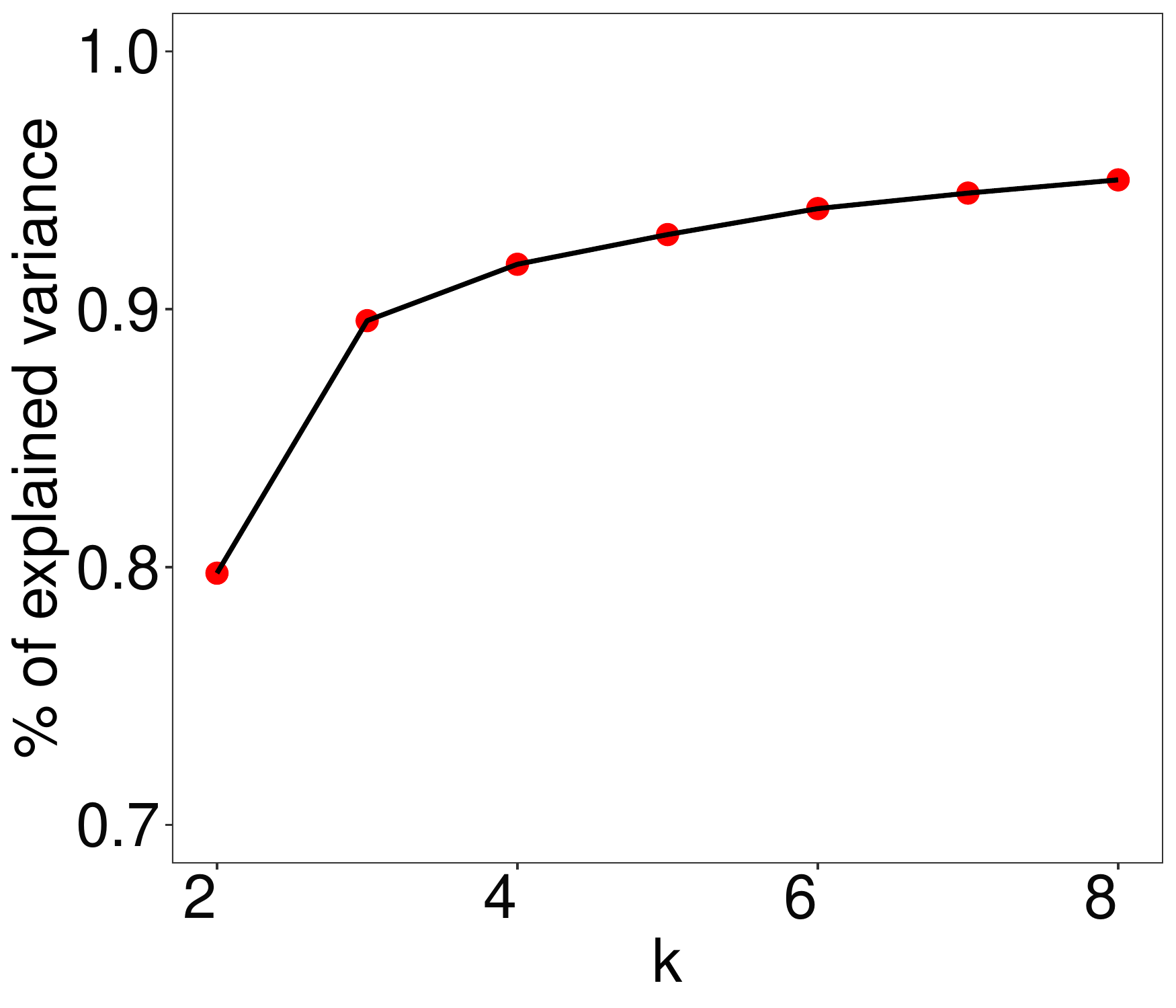}
			\caption{The CLR of $504$ S\&P 500 stocks from December 2011 to December 2021 (left). Variances of $\X$ explained by the $k$ archetypes identified by AA as a function of $k$ for $k=2,\cdots, 8$ (right).}
			\label{fig:1}
		\end{center}
	\end{figure}
	
To apply AA, we need to first determine the number of archetypes $k$. Like other unsupervised learning methods such as the $k$-means and PCA, there is no principled rule to find the correct number of $k$ for real-life datasets. 
	A more practical solution is to follow the heuristic ``elbow rule'' \cite{thorndike1953belongs} to choose $k$ approximately. In this case, we apply SVD-AA to find such a $k$. In particular, we plot the variance of the dataset explained by the archetypes given by SVD-AA as a function of $k$ (see Figure \ref{fig:1}) and choose $k$ to be the point where the curve starts to plateau, which is around $k=3$.
	
	Setting $k=3$, we apply SVD-AA, AAA, and archetypes to compute the archetypes for $\X$. 
	The parameters $p$, $M$ and $\eta$ in AAA are set as $20$, $10000$ and $0.003$ (so that $\eta/3 = 0.001$), respectively.
	Each experiment is repeated $100$ times, with the learned archetypes (in the first 10 experiments), the running times (elapsed time computed using the `\texttt{system.time()}' function in \texttt{R}) and residuals reported in Figure \ref{fig:2} and Figure \ref{fig:3}, respectively.
	
	\begin{figure}[htbp]
		\begin{center}
			\includegraphics[width=0.193\textwidth]{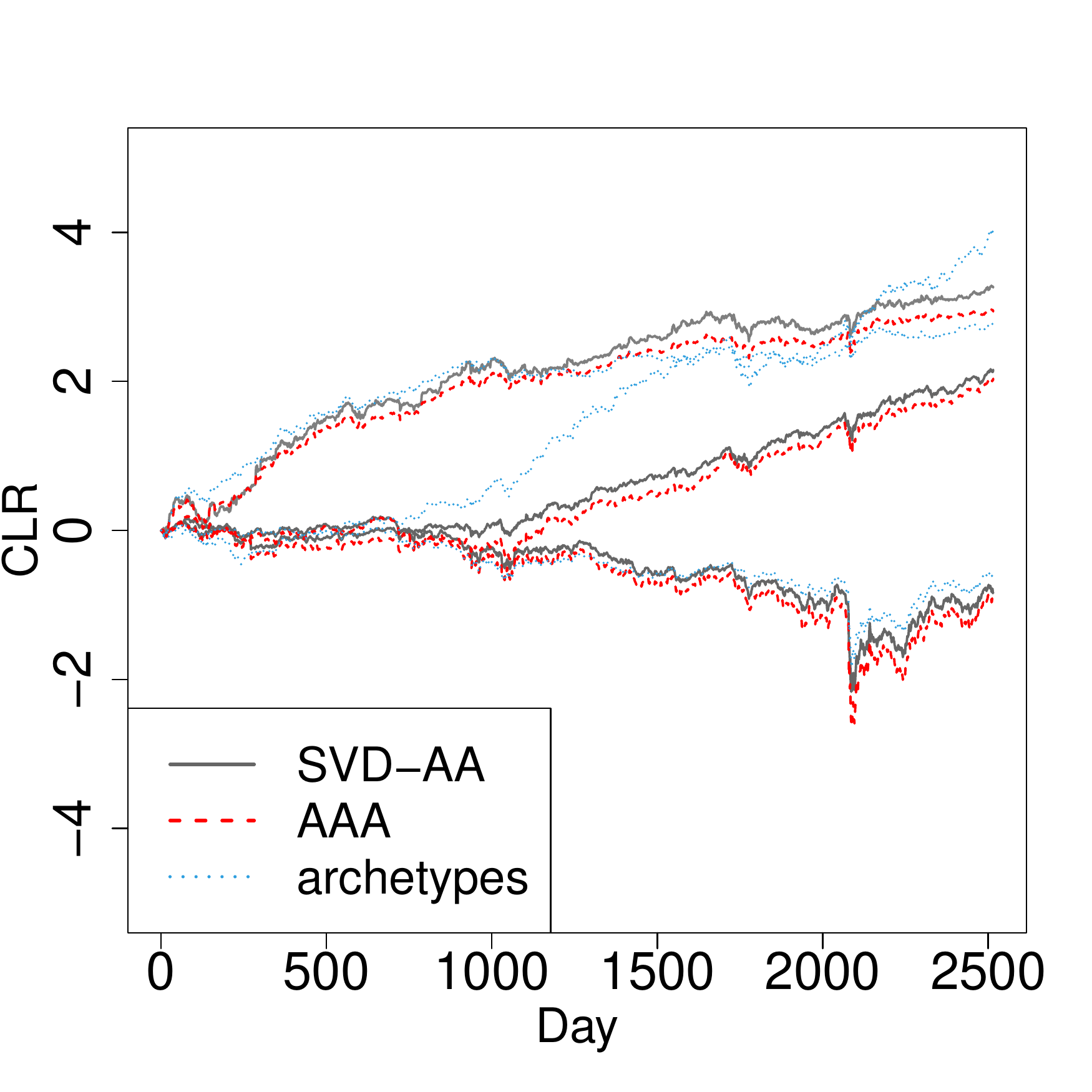}
			\includegraphics[width=0.193\textwidth]{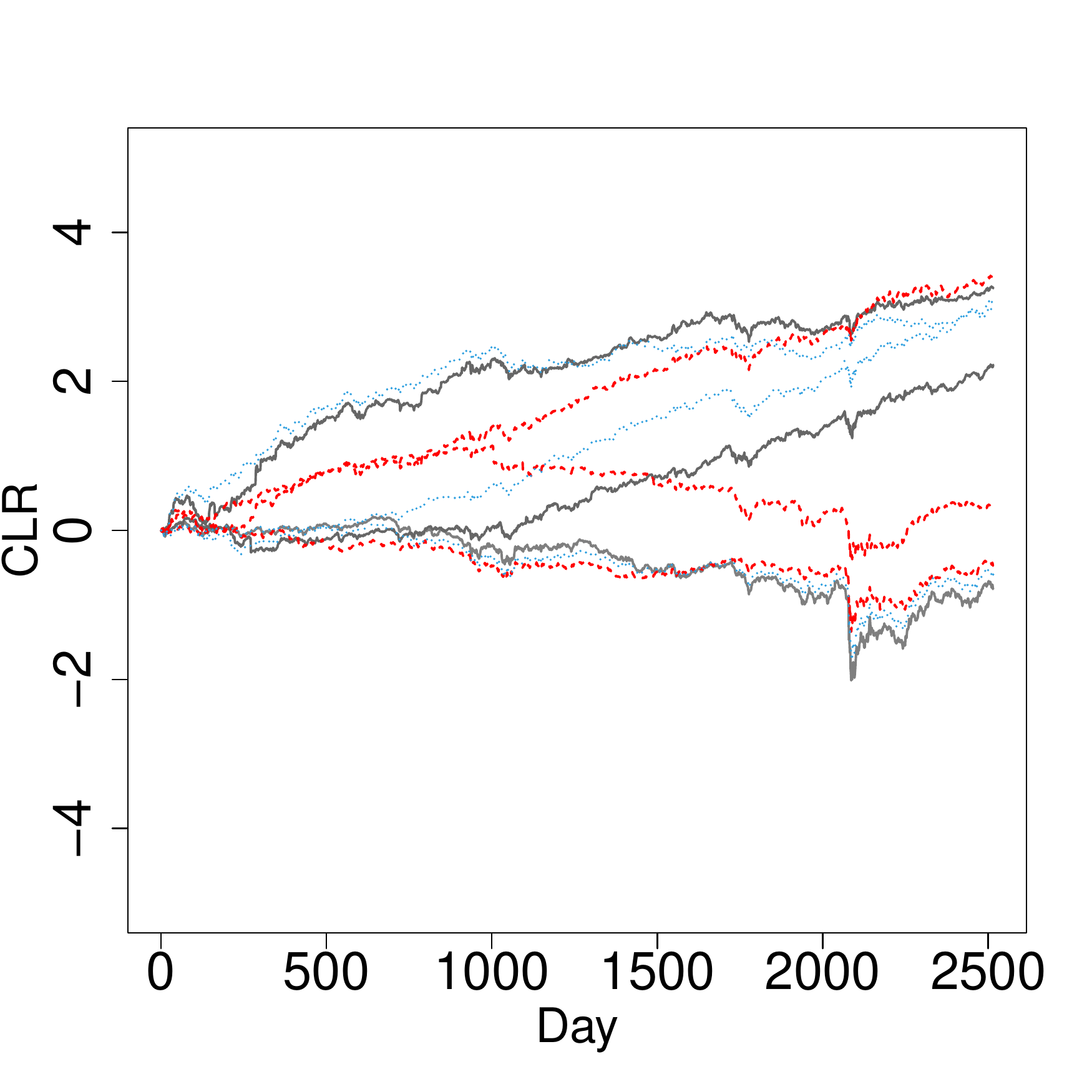}
			\includegraphics[width=0.193\textwidth]{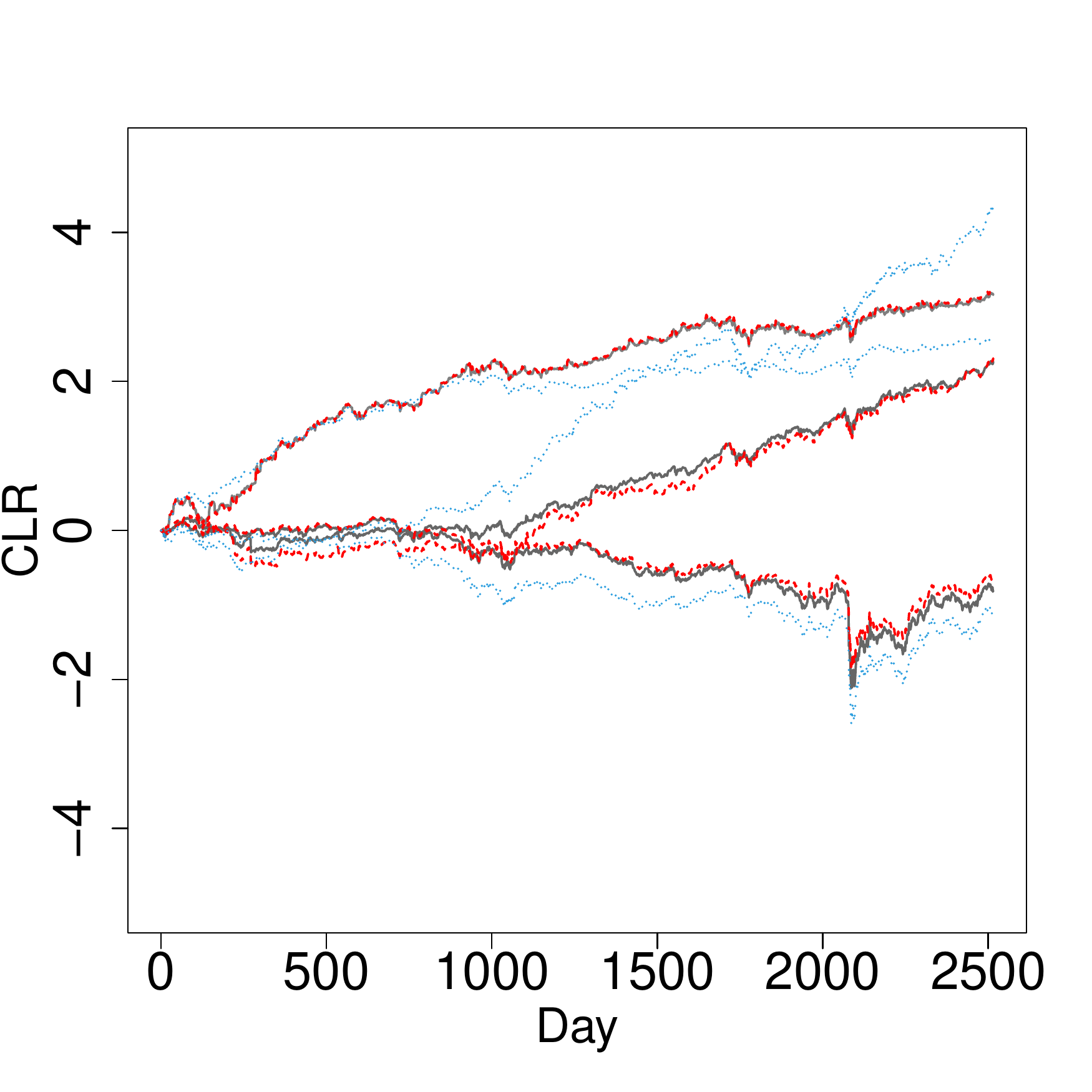}
			\includegraphics[width=0.193\textwidth]{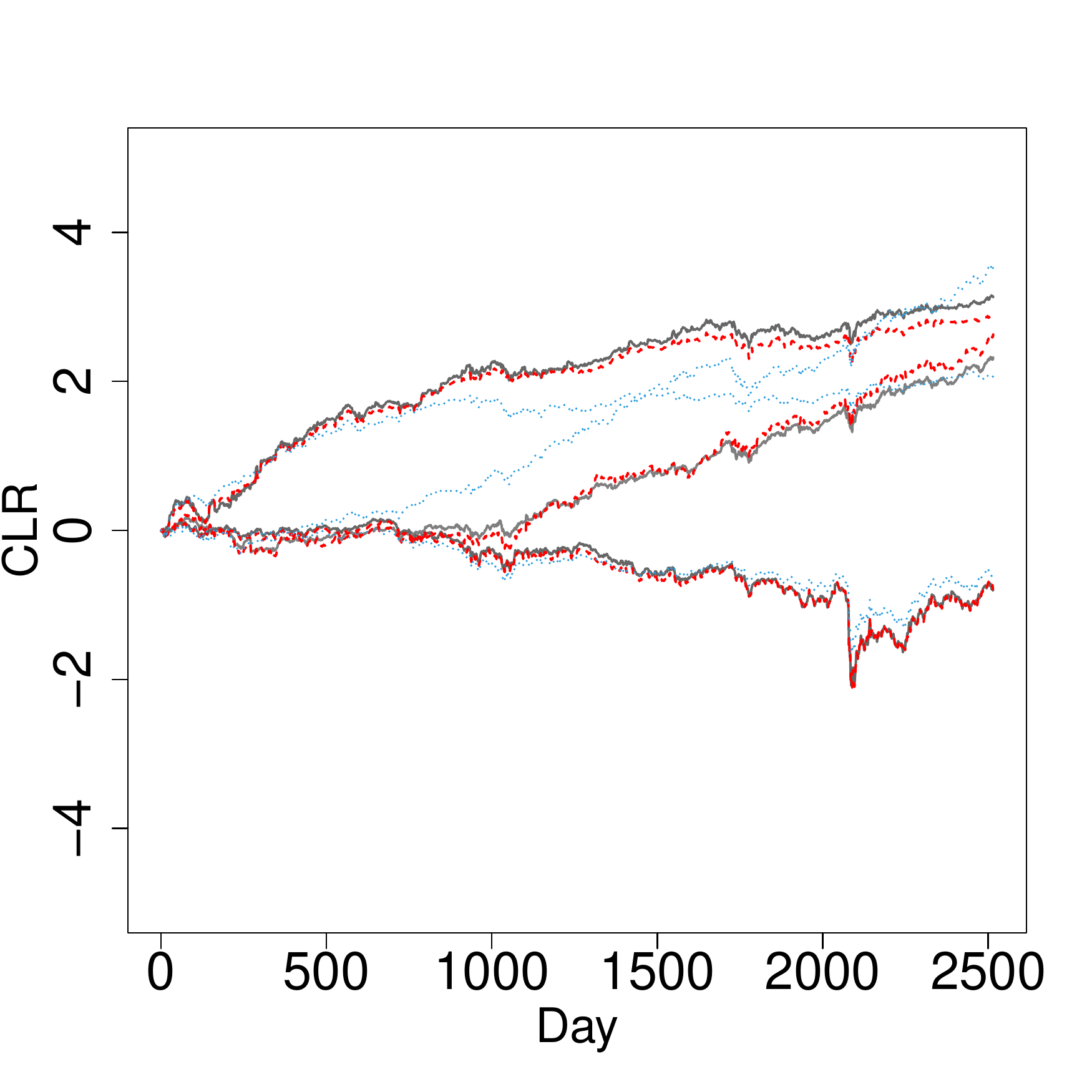}
			\includegraphics[width=0.193\textwidth]{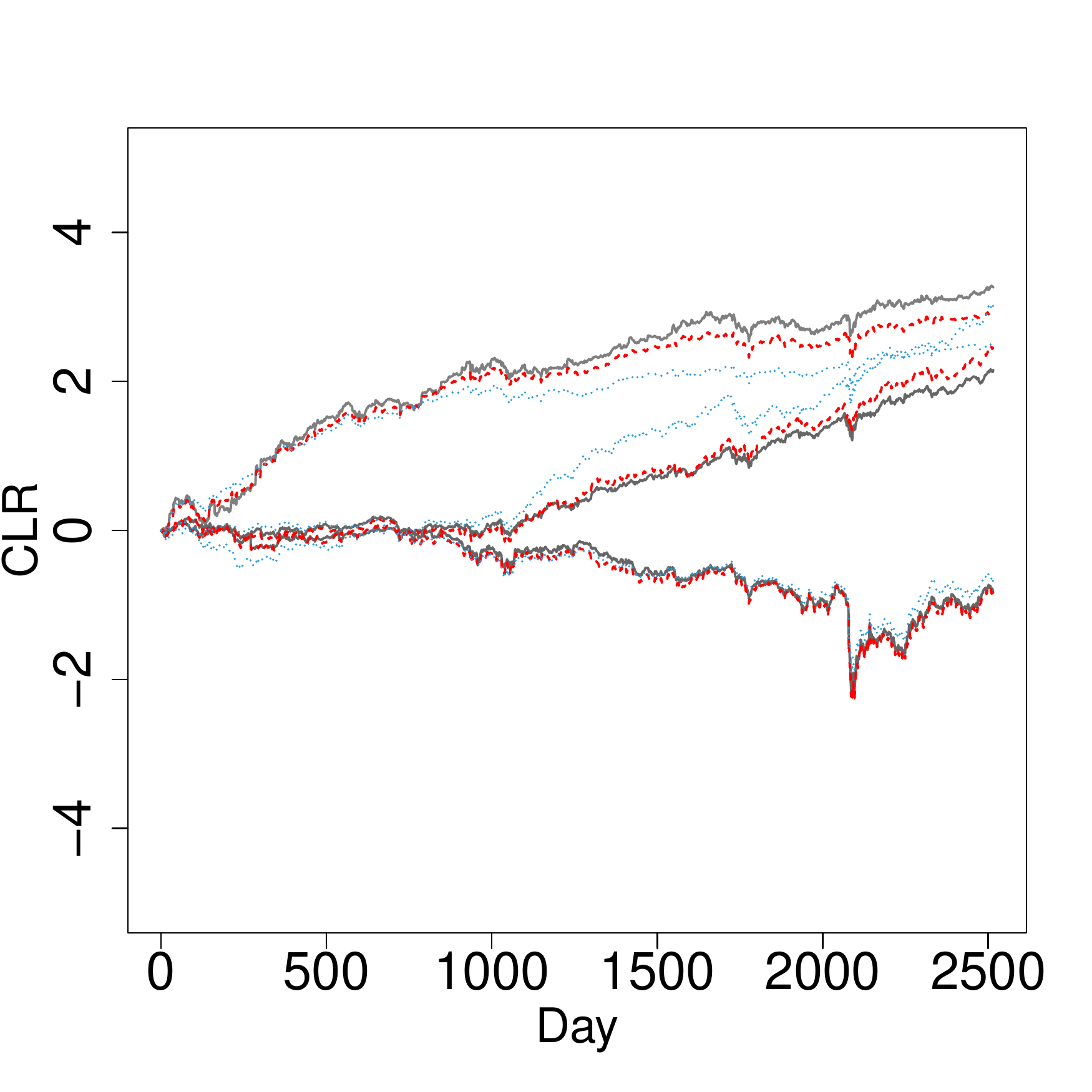}
			\includegraphics[width=0.193\textwidth]{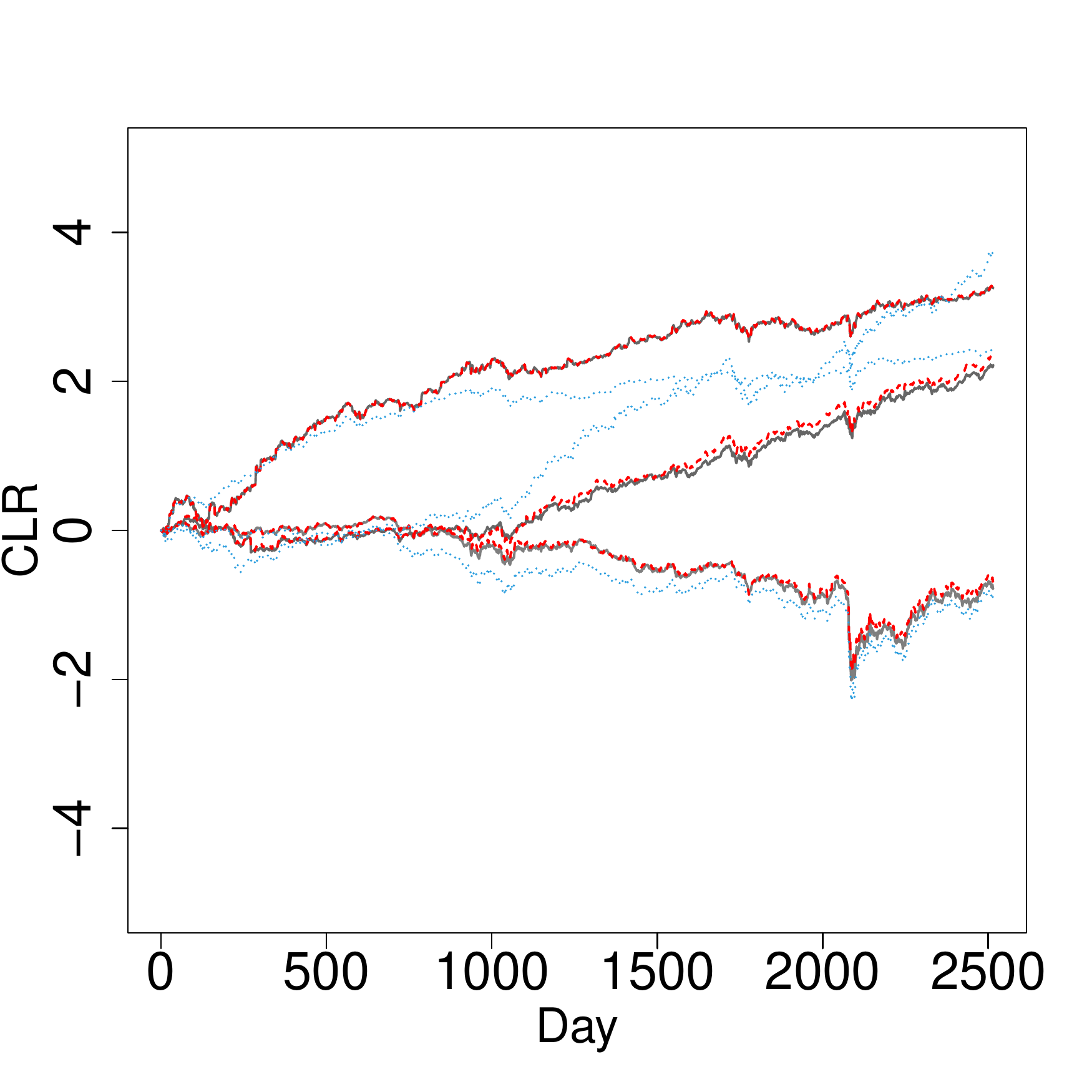}
			\includegraphics[width=0.193\textwidth]{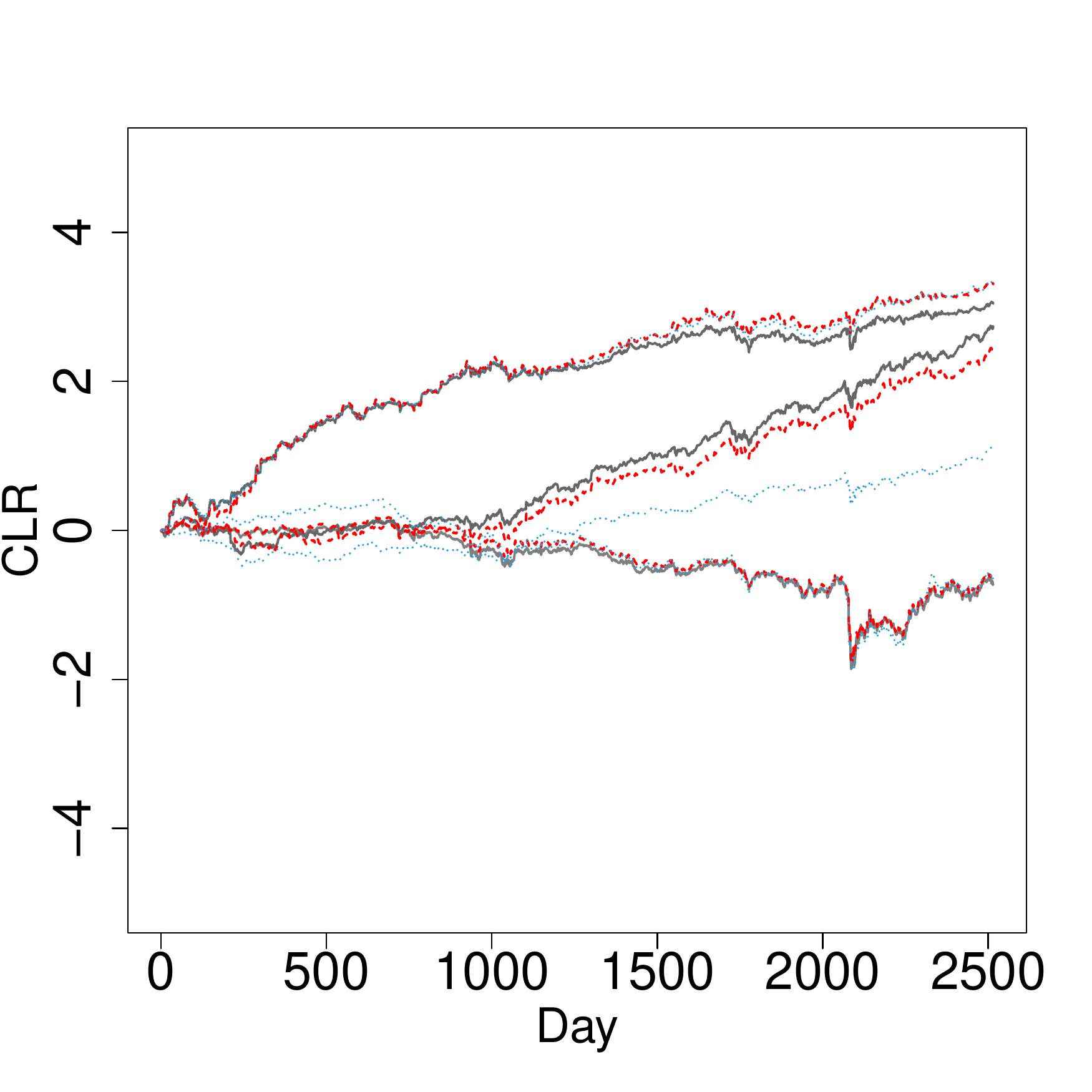}
			\includegraphics[width=0.193\textwidth]{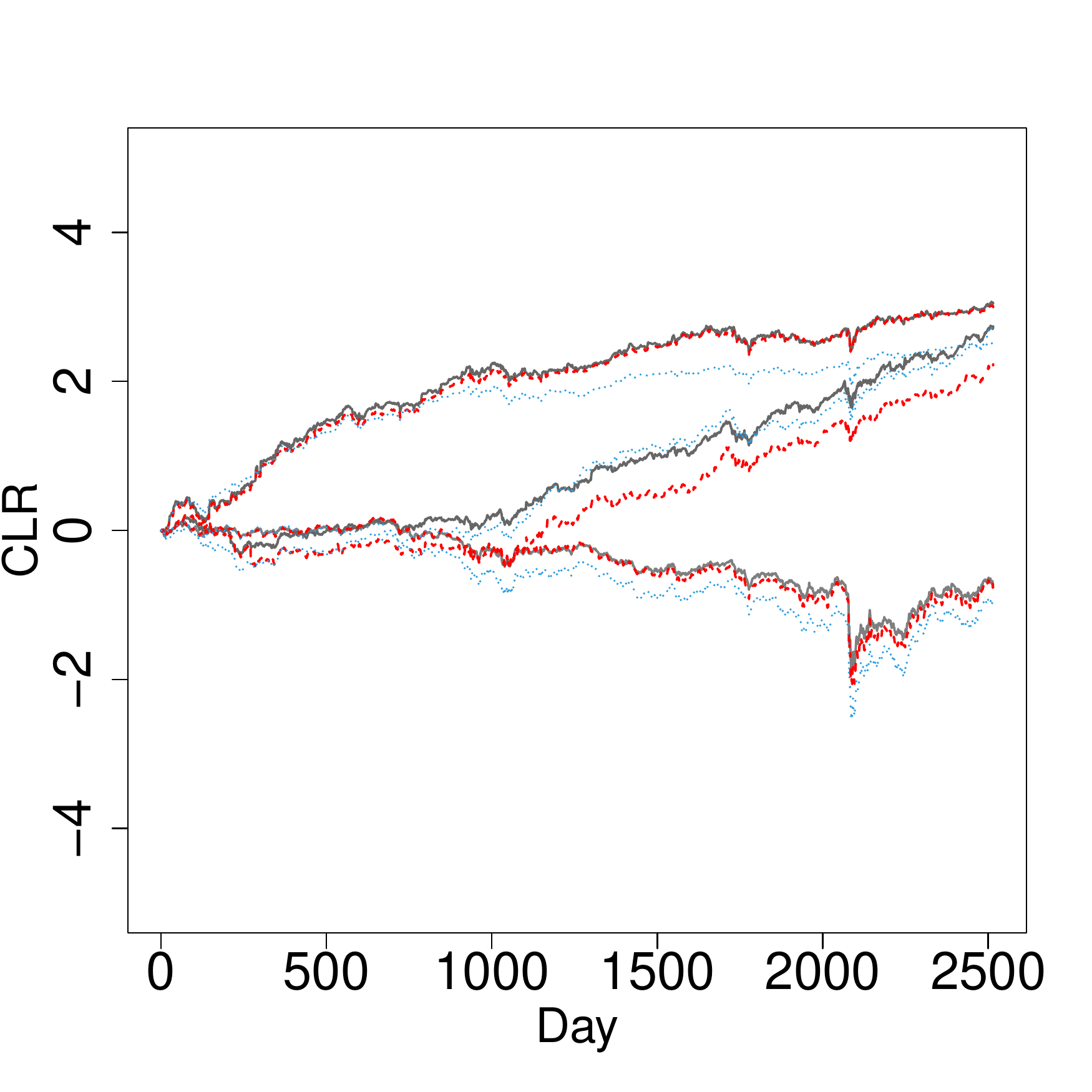}
			\includegraphics[width=0.193\textwidth]{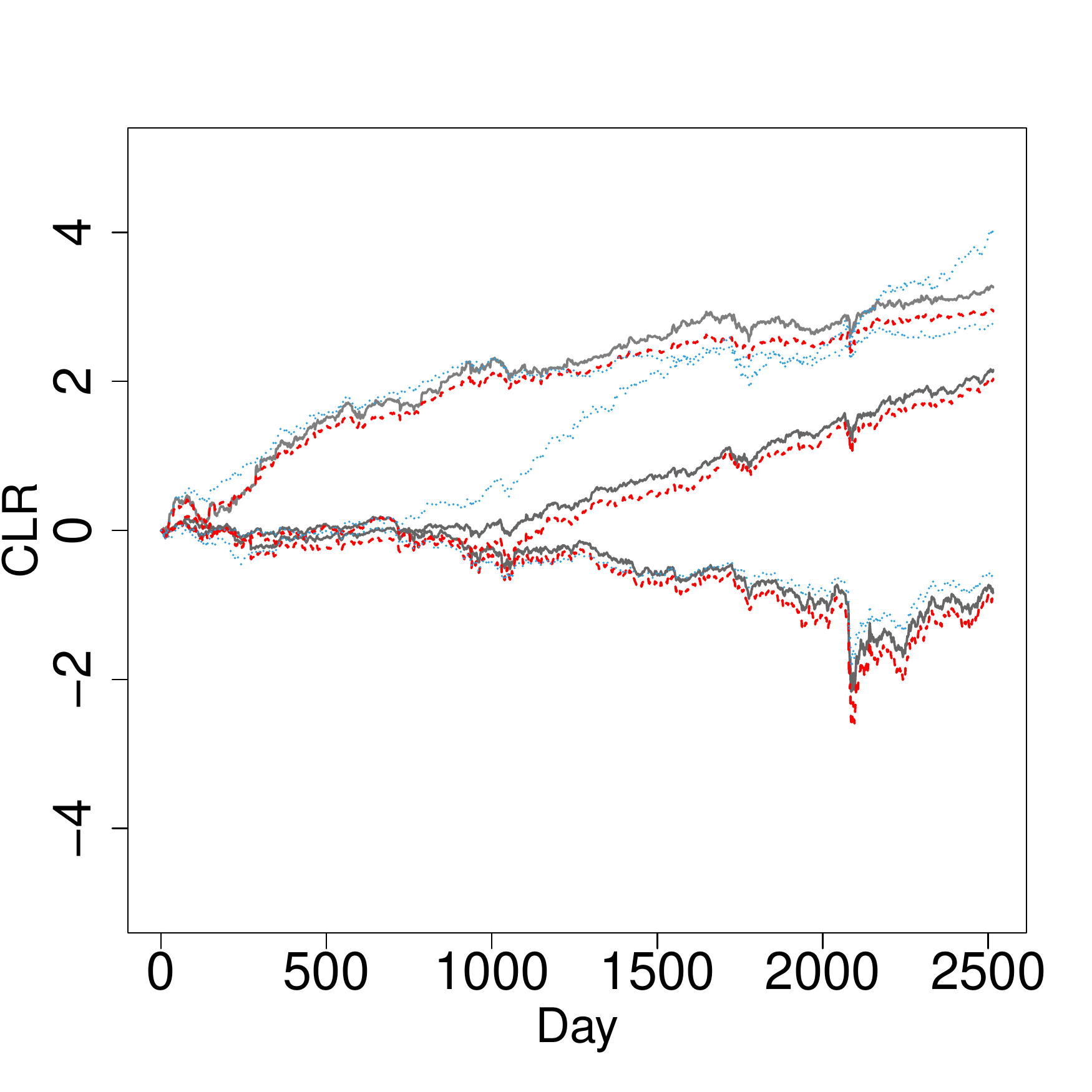}
			\includegraphics[width=0.193\textwidth]{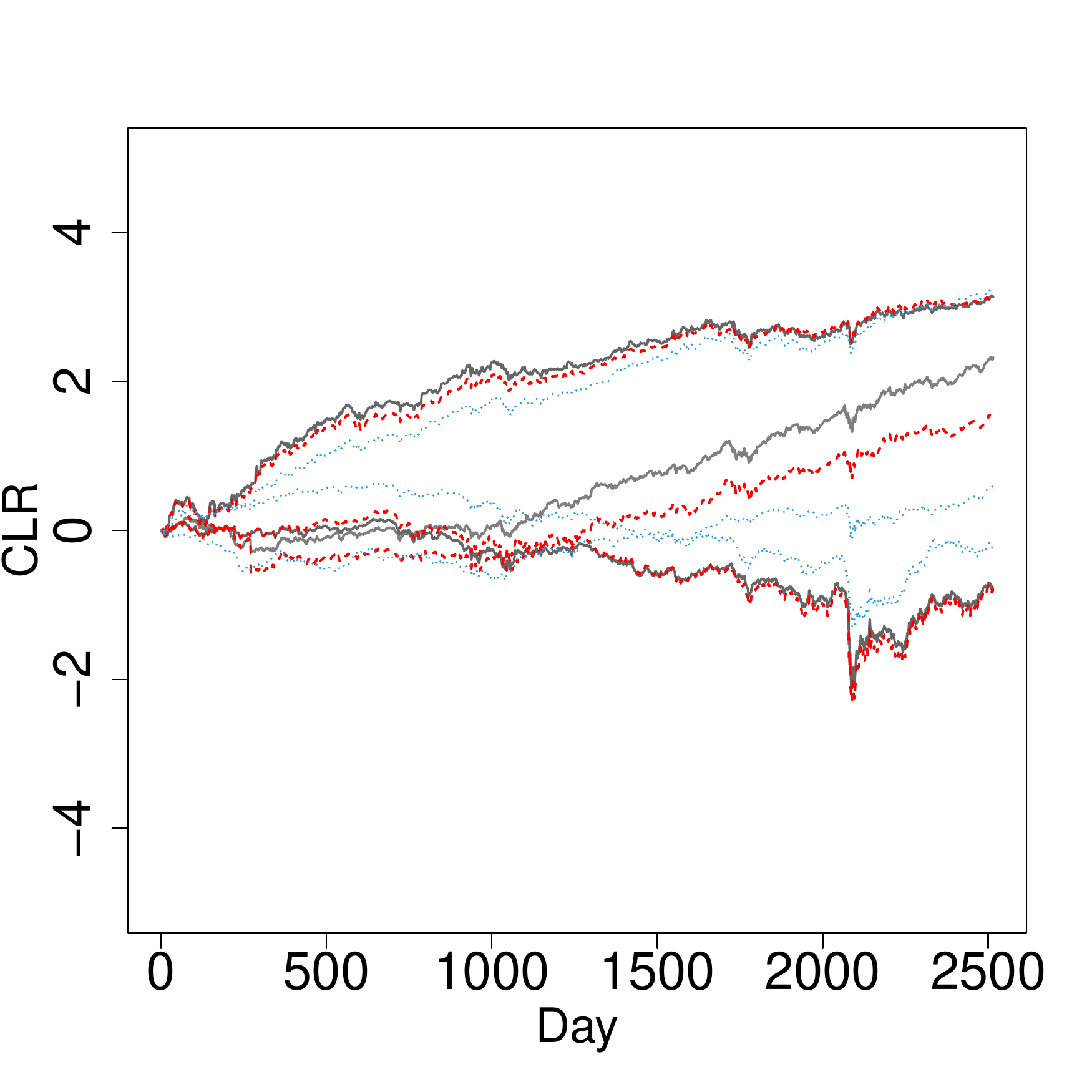}
			\caption{Instances of the computed archetypes by SVD-AA, AAA and archetypes in the first 10 experiments. }
			\label{fig:2}
		\end{center}
	\end{figure}
	
	\begin{figure}[htbp]
		\begin{center}
			\includegraphics[width=0.40\textwidth]{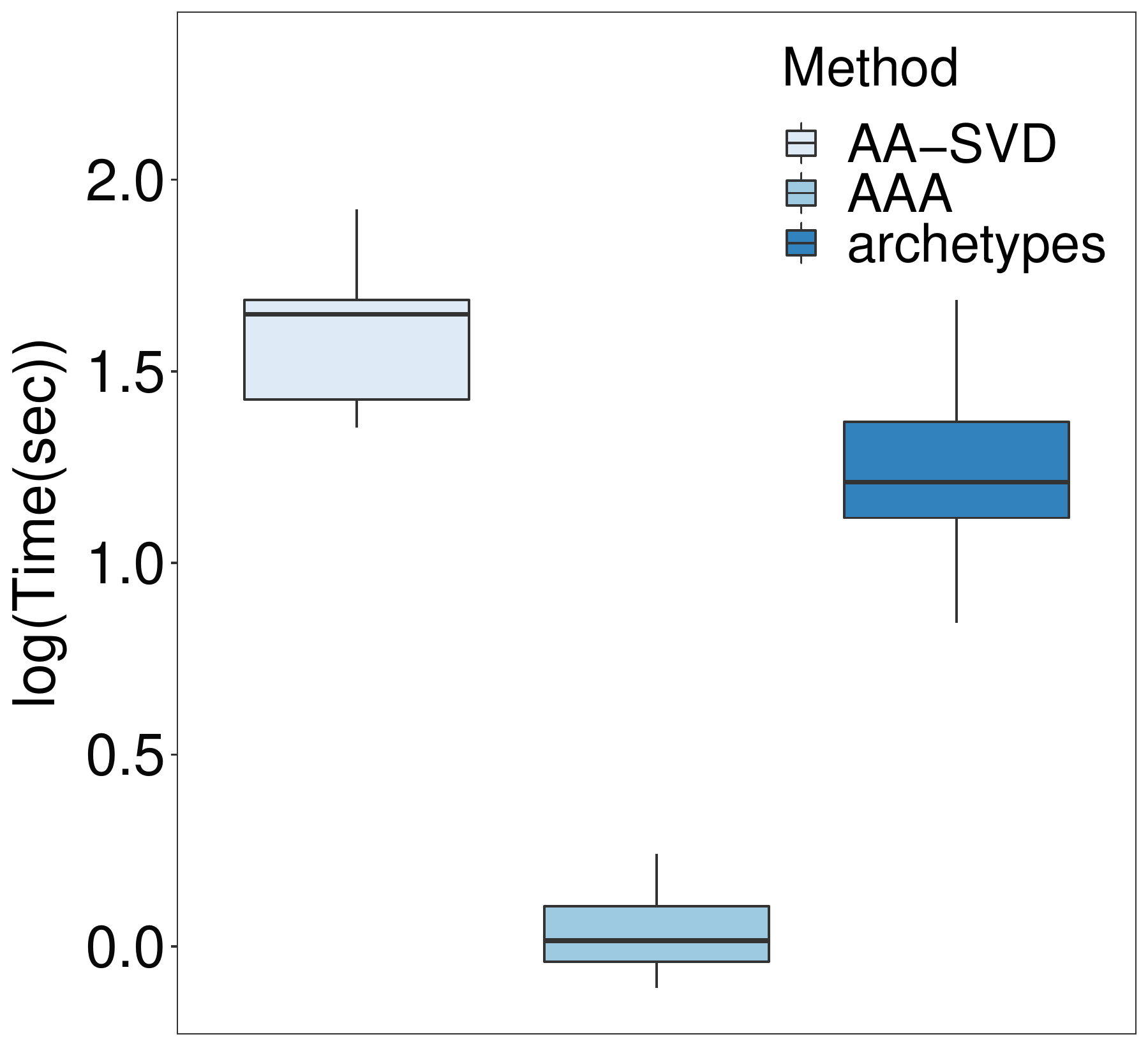}\hspace{2 cm}
			\includegraphics[width=0.40\textwidth]{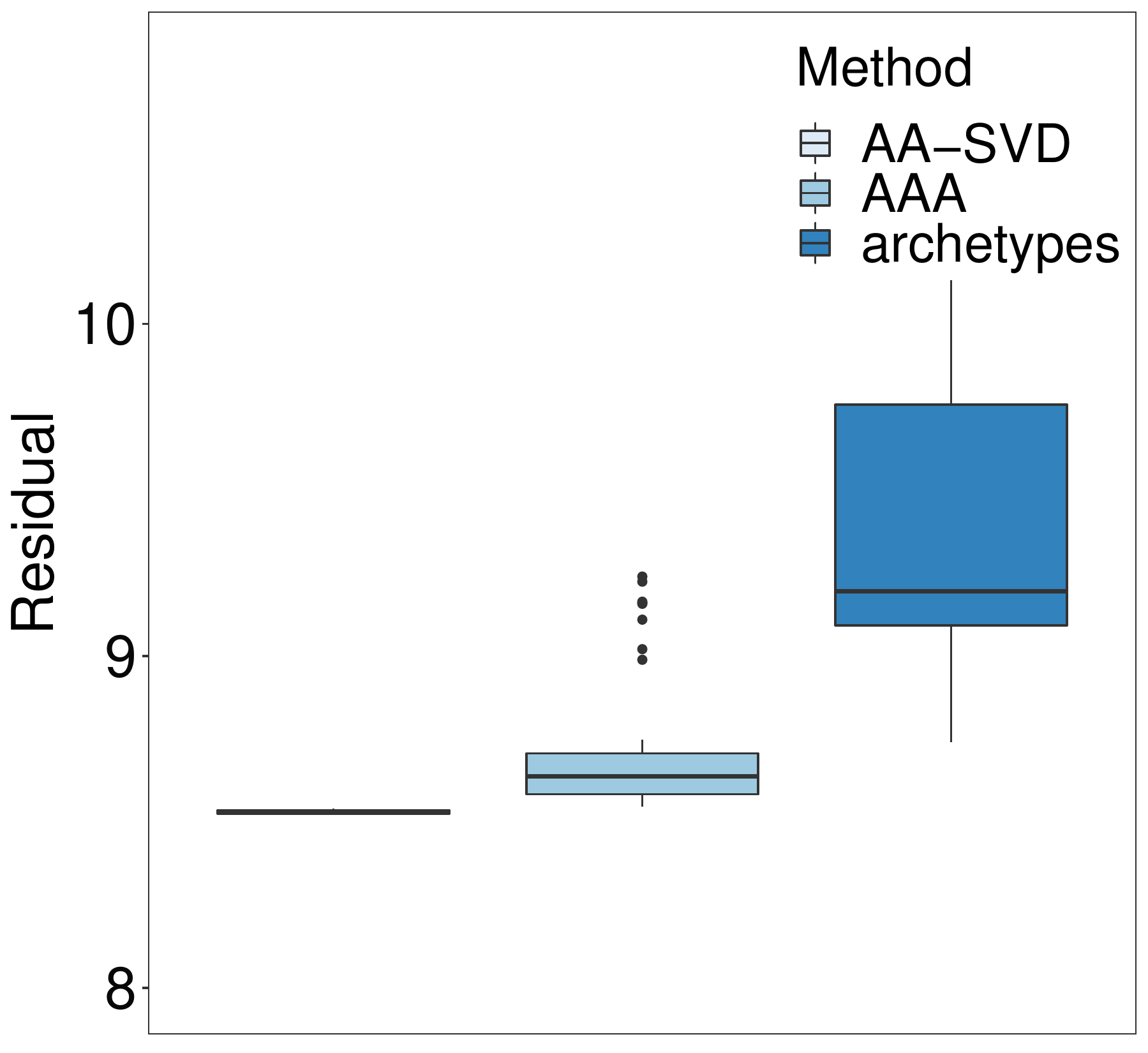}
			\caption{Boxplots of the running times (left) and residuals (right) of SVD-AA, AAA and archetypes in 100 experiments. }
			\label{fig:3}
		\end{center}
	\end{figure}
	
	\begin{figure}[htbp]
		\begin{center}
			\includegraphics[width=0.40\textwidth]{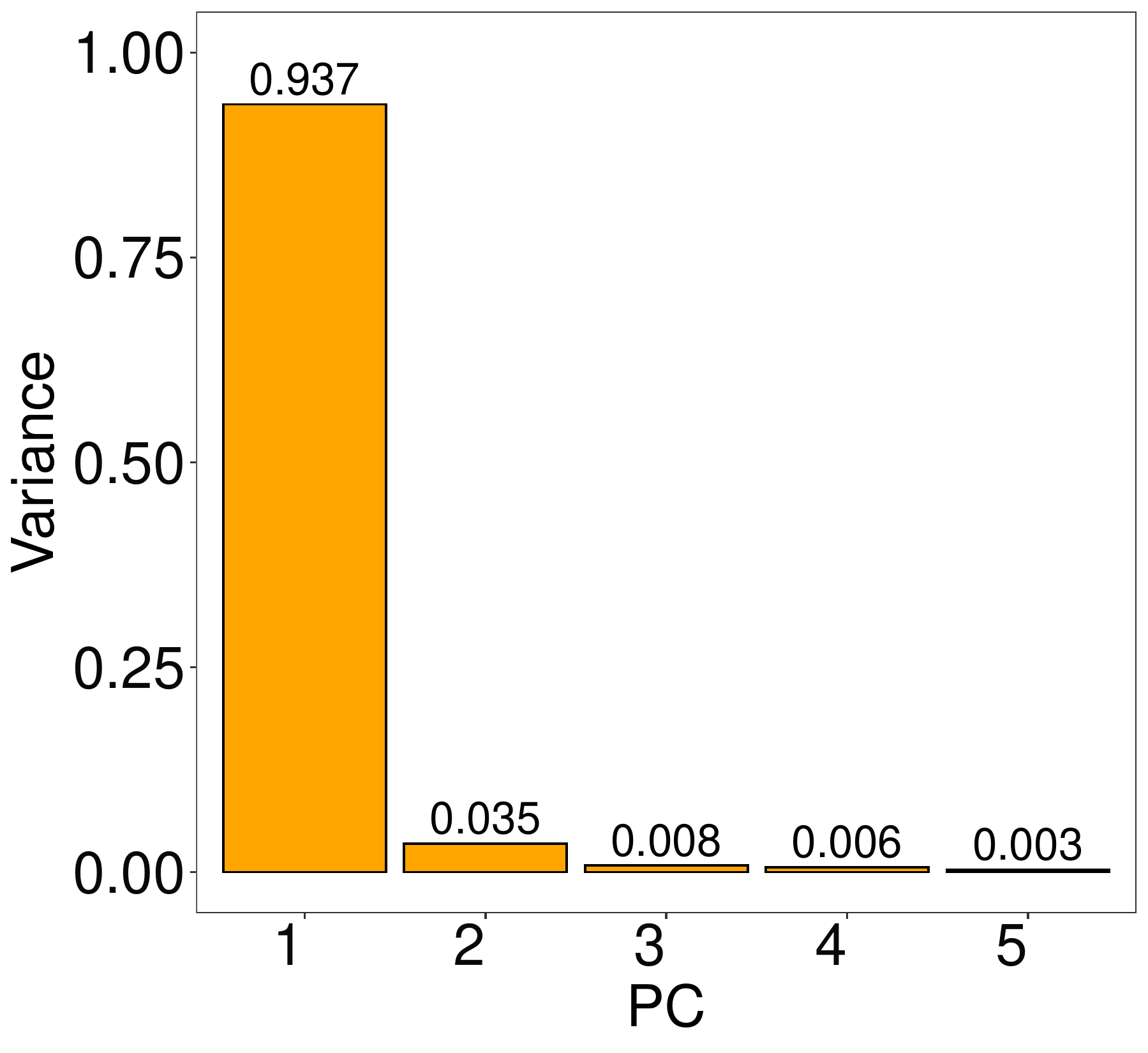}\hspace{2 cm} 
			\includegraphics[width=0.40\textwidth]{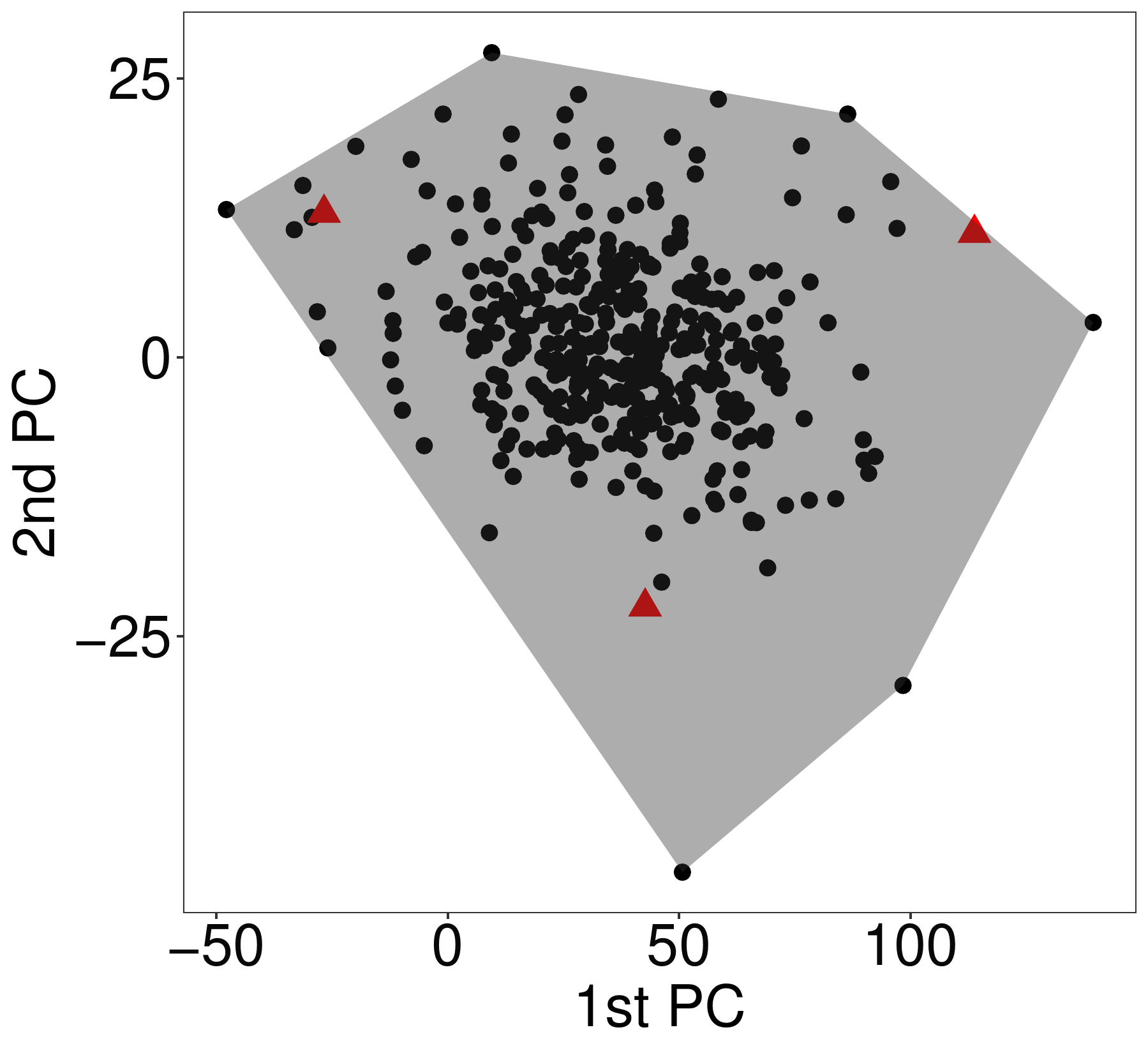}
			\caption{Variances explained by the first five principal components of $\X$ (left). Scatterplot of the reduced representation of $\X$ with respect to the first two left singular vectors (which account for $97\%$ of the variation of the dataset) and its convex hull. The red triangles are the reduced representation of the three archetypes (right). }
			\label{fig:4}
		\end{center}
	\end{figure}
	
	It can be seen from Figure \ref{fig:3} that SVD-AA, as expected, gives the best-computed archetypes in terms of the residual on average; however, its computation time is significantly longer than the other two methods. 
	The built-in function archetypes has the worst performance, and its computation time is between the other two methods. 
	The AAA, which first reduces the dimension of the dataset before applying the alternating minimization, achieves competitive results with SVD-AA (despite a few outliers) but takes much less time (more than $30$ times faster than SVD-AA). 
	This may be because $\X$ is essentially low-dimensional and admits a parsimonious approximation for its convex hull. A numerical justification for this argument can be seen from the spectral decay of the sample covariance matrix of $\X$ as well as the scatterplot of the reduced representation of $\X$ with respect to the first two principal components (PCs), as illustrated in Figure \ref{fig:4}.  
	
	To implement AAA, it is necessary to choose the input parameters in advance.
	The optimal choice for the parameters is problem-dependent and often there is no universal tuning strategy for it.
	The Krylov subspace parameter $s$ is set as $\lceil\log N\rceil$ deterministically.   
	We investigate the accuracy/running time dependence on $p, M$, and $\eta$.  
	In particular, we will use the same parameters as in the previous simulation.
	Whenever we test the dependence on one parameter, the other two are set fixed.   
	We will test $p, M$ and $\eta$ at three different values, respectively, i.e., $p = 10, 20, 30$, $M = 10^3, 10^4, 10^5$ and $\eta = 0.3, 0.03, 0.003$. The results are given in Figure \ref{fig:5}. 
	
	\begin{figure}[htbp]
		\begin{center}
			\includegraphics[width=0.40\textwidth]{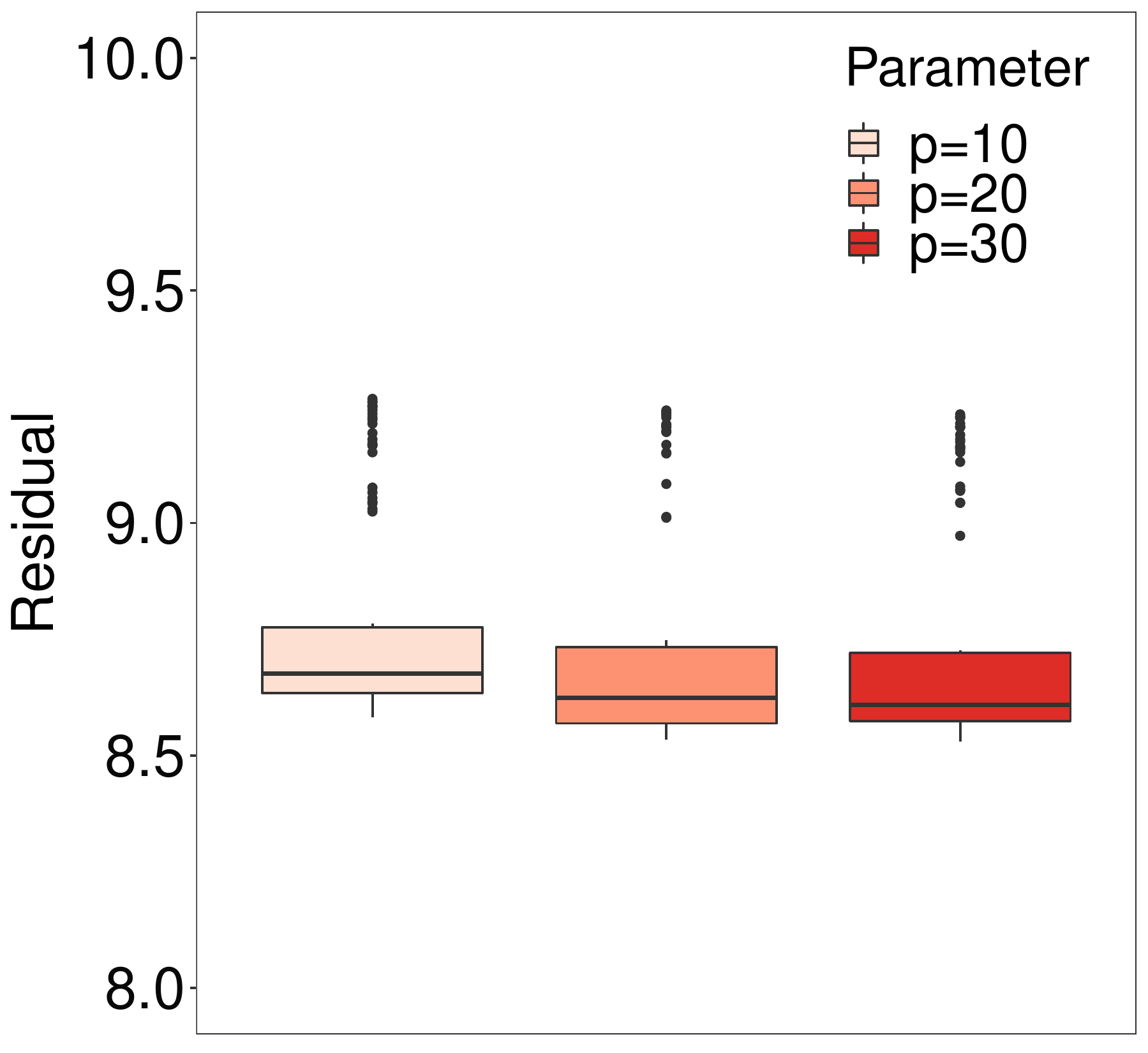}\hspace{2 cm}
			\includegraphics[width=0.40\textwidth]{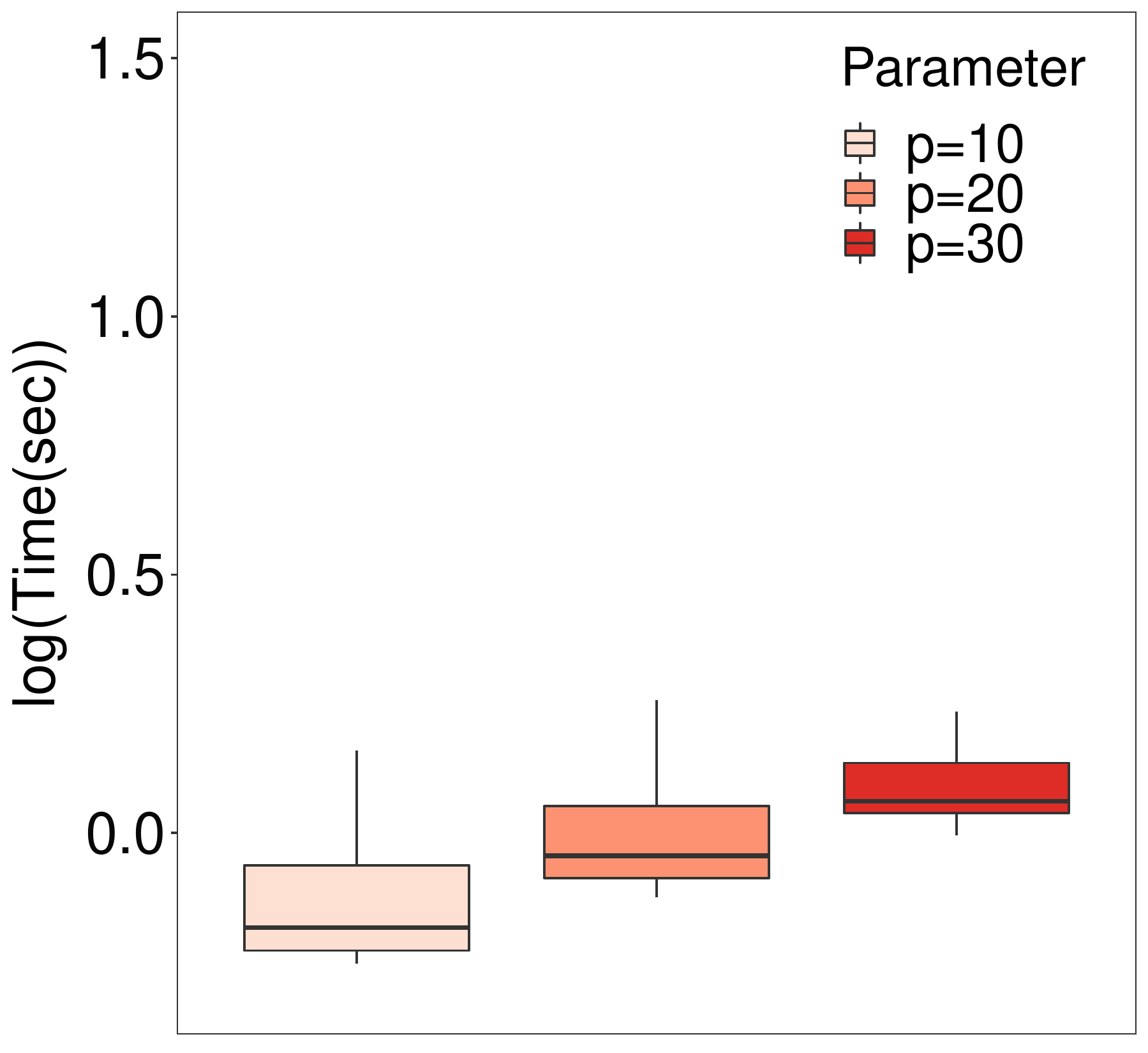}
			\vspace{0.4 cm}
			\includegraphics[width=0.40\textwidth]{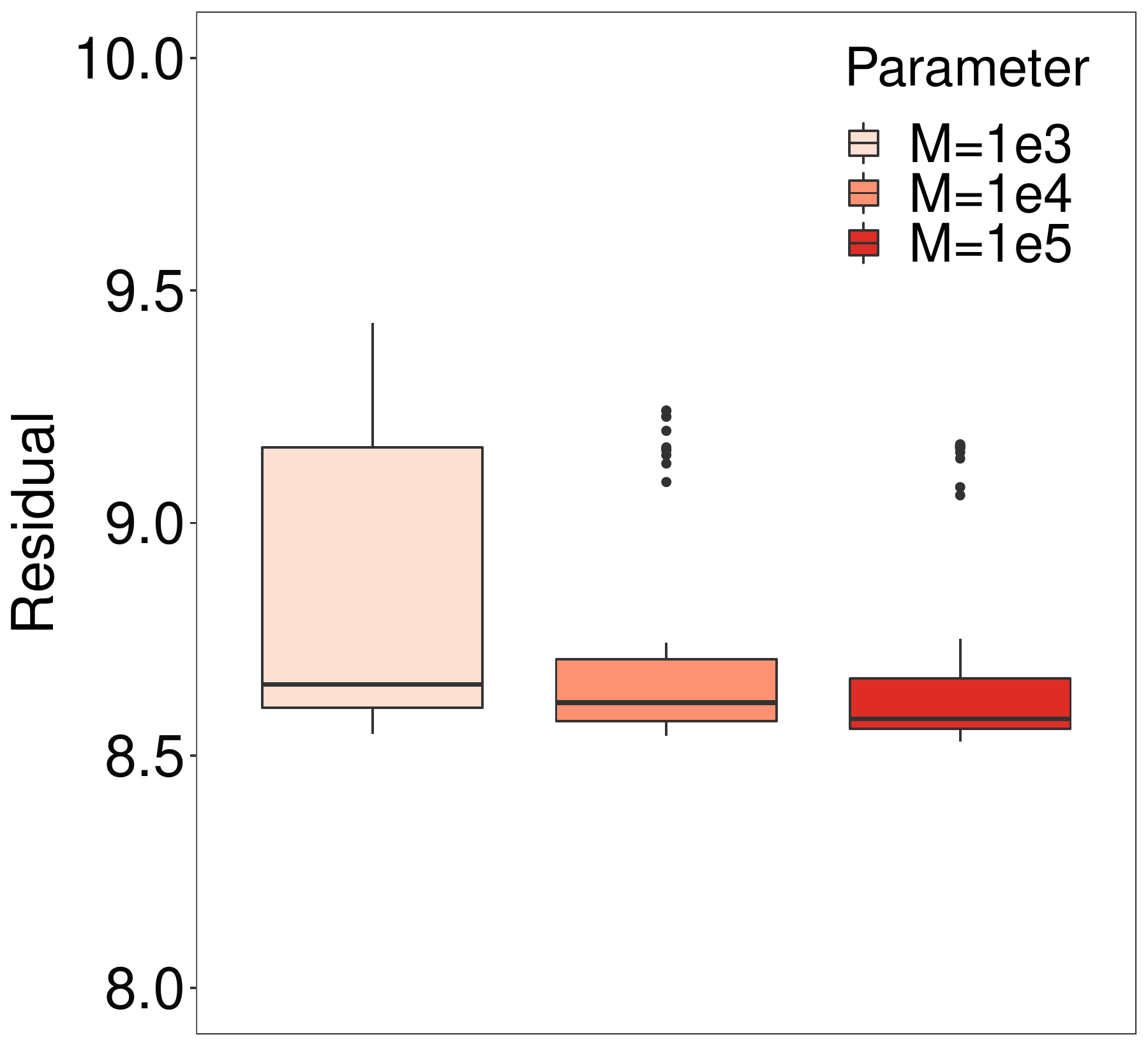}\hspace{2 cm}
			\includegraphics[width=0.40\textwidth]{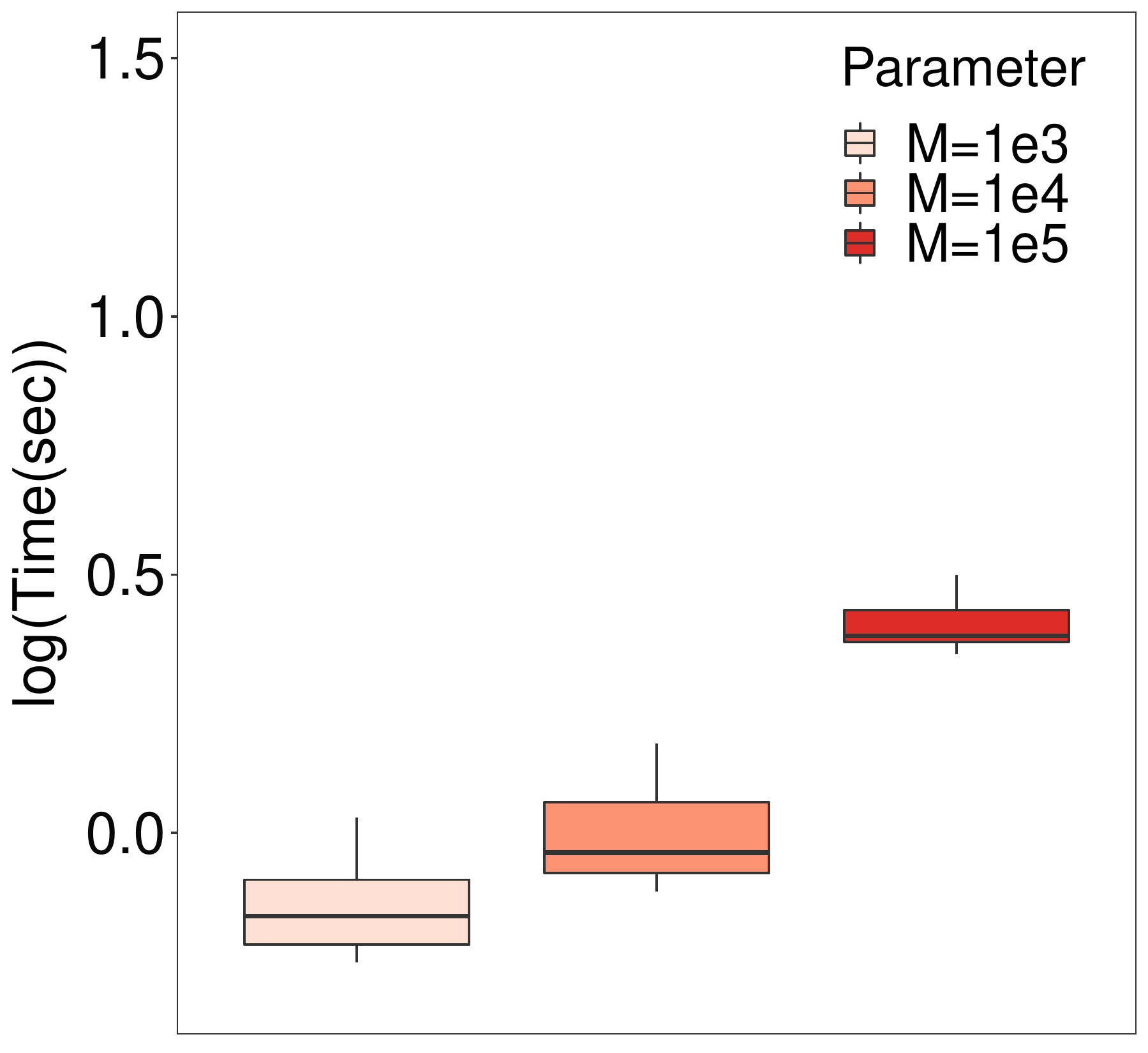}
			\vspace{0.4 cm}
			\includegraphics[width=0.40\textwidth]{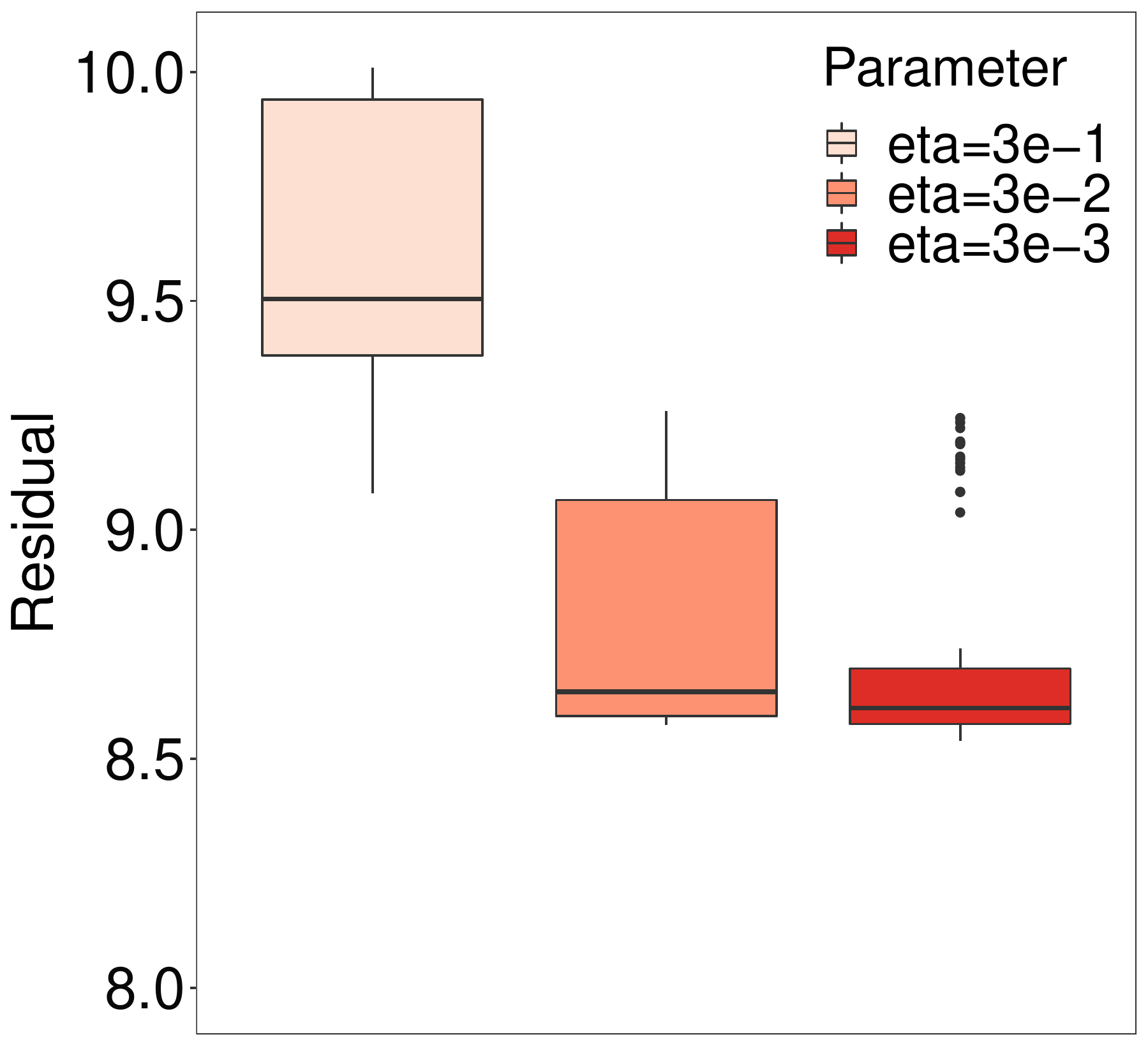}\hspace{2 cm}
			\includegraphics[width=0.40\textwidth]{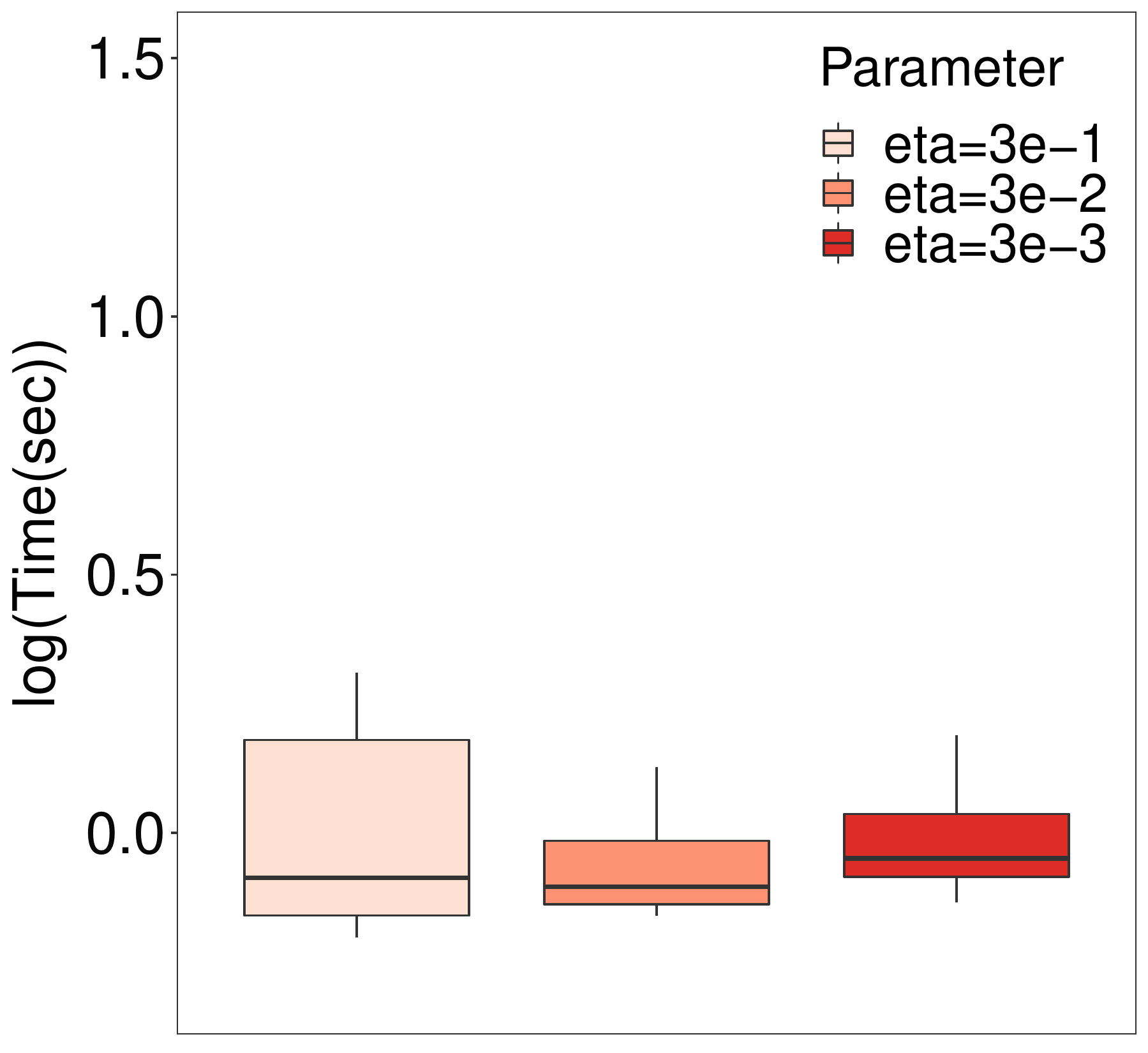}
			\caption{Accuracy/running time dependence of AAA algorithm for the S\&P 500 dataset with base parameters $p=20, M = 10^4$ and $\eta = 0.003$.}
			\label{fig:5}
		\end{center}
	\end{figure} 
	
	Figure \ref{fig:5} shows that for the S\&P 500 dataset, the accuracy of AAA has a strong dependence on $\eta$, which measures the missing proportion of curvature in the approximate convex hull construction. 
	The number of random projections $M$ mostly influences the running time while having only a mild impact on the accuracy when $p=20$. 
	The approximation rank $p$, as long as set reasonably large, is sufficient to give a good approximation result.  
	
	In this example, we compare the three archetypes with the same number of centers identified by the $k$-means; see the first plot in Figure \ref{fig:44}.
	It can be seen that the archetypal curves are visually more illustrative than the centers of the $k$-means, which share a similar growth pattern with differing slopes.     
	Indeed, the percentage of variance explained by AA is around $90\%$; the same number for the $k$-means and PCA are $64\%$ and $98\%$, respectively. 
	Visualization of the convex coefficients for each data point with respect to the three archetypes is given in the ternary plot in Figure \ref{fig:44}. 
	In this case, most of the data fall in the interior of the simplex, suggesting the S\&P dataset can be well summarized using a polytopic structure. 		
		\begin{figure}[htbp]
		\begin{center}
			\includegraphics[width=0.40\textwidth]{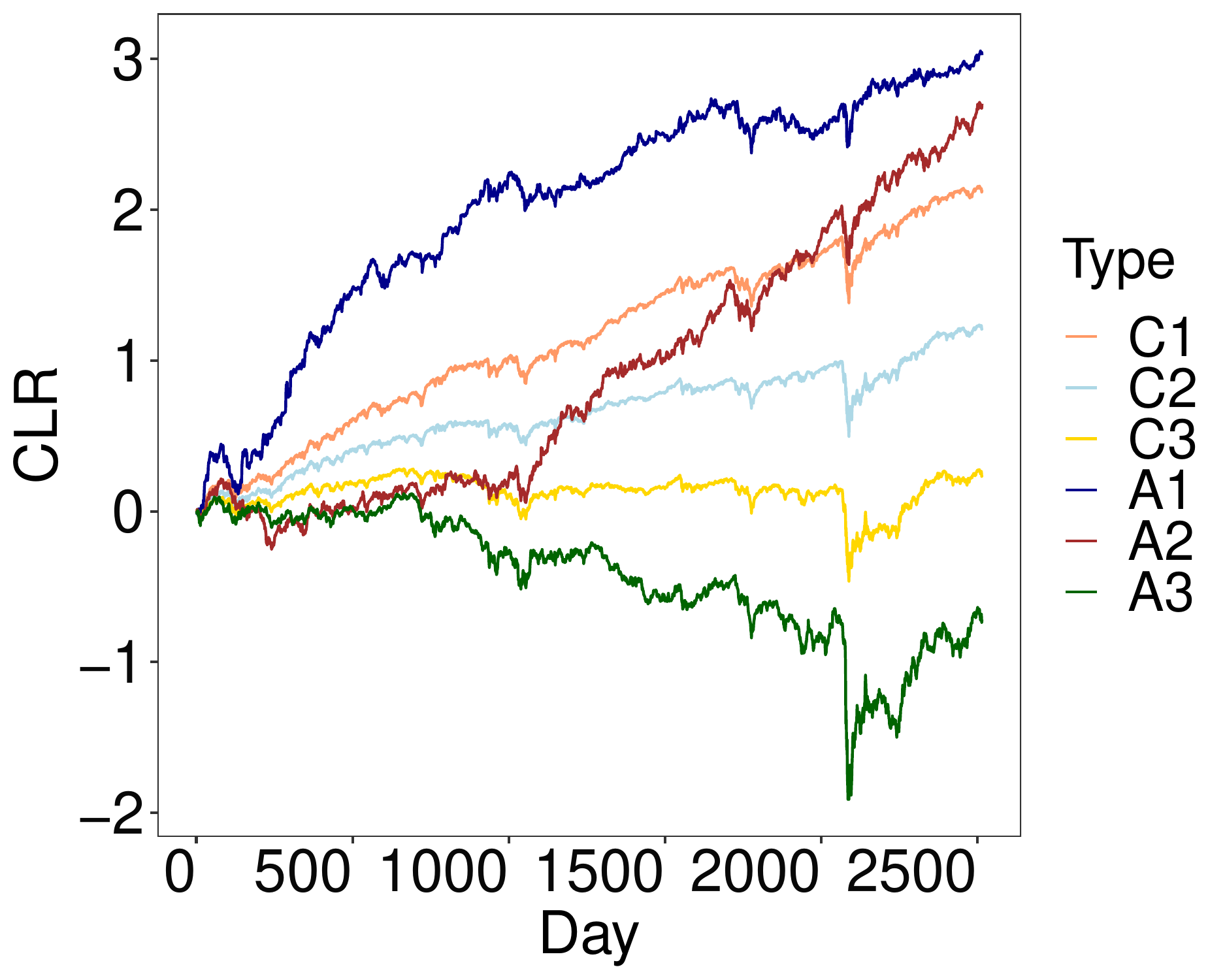}\hspace{2 cm} 
			\includegraphics[width=0.40\textwidth]{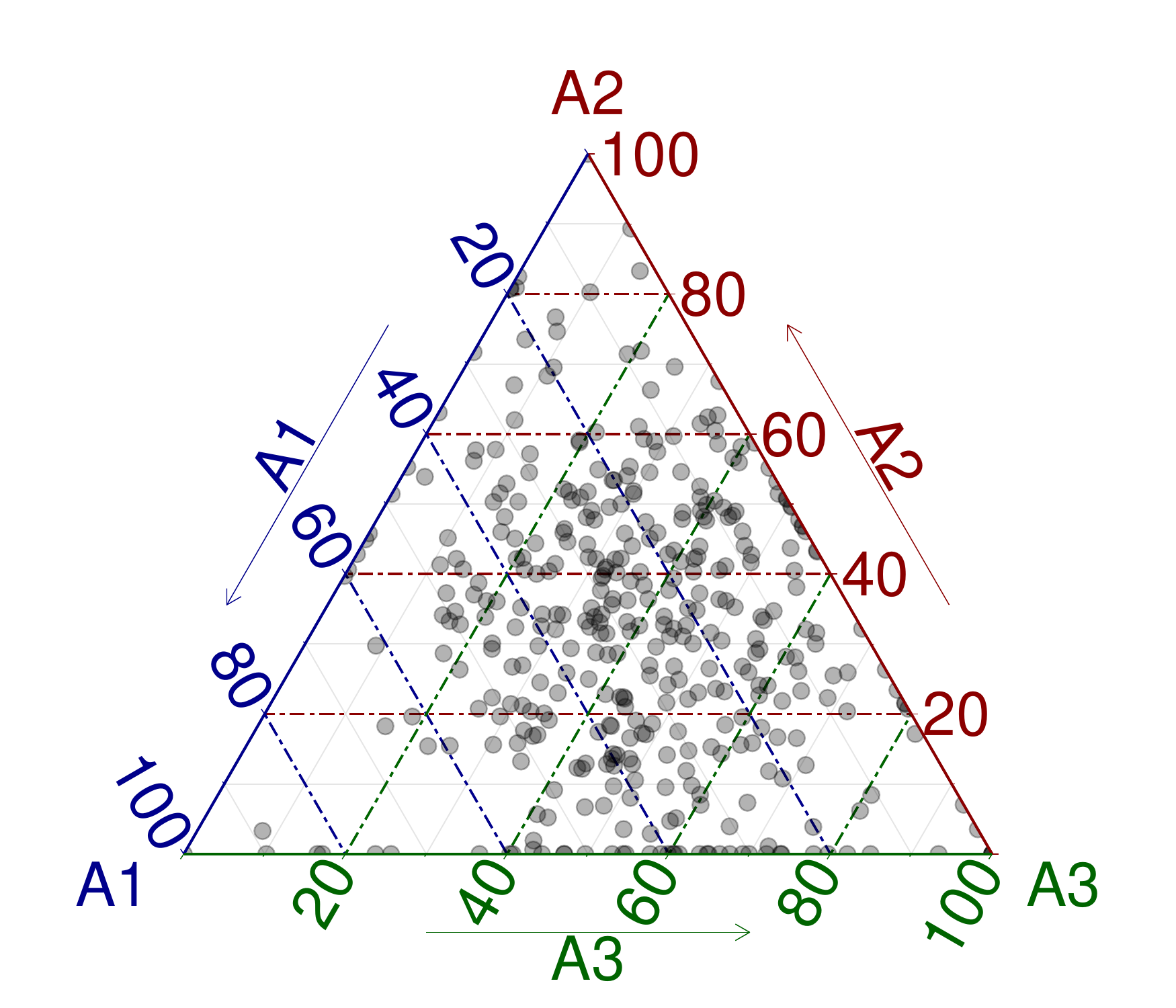}
			\caption{Comparison of the three centers (C1, C2, C3) given by the $k$-means and the three archetypes (A1, A2, A3) given by the SVD-AA for the S\&P 500 dataset (left). Visualization of the convex combination coefficients of each data point with respect to the three archetypes (right).}
			\label{fig:44}
		\end{center}
	\end{figure}

To further understand the meanings of the three archetypes, for each archetype, we single out the tickers of the top five companies having the largest component in the direction:
\begin{itemize}
\item A1: NFLX: $(\bm{1}, 0, 0)$, STZ.B: $(\bm{0.9},0.1,0)$, ILMN: $(\bm{0.89}, 0.03, 0.08)$, REGN: $(\bm{0.84}, 0.16, 0)$, FLT: $(\bm{0.83}, 0.17, 0)$;
\item A2: AMD: $(0, \bm{1}, 0)$, FTNT: $(0, \bm{0.89},0.11)$, ISRG: $(0.02, \bm{0.83}, 0.15)$, LRCX: $(0.17, \bm{0.83}, 0)$, CPRT: $(0.18, \bm{0.81}, 0)$;
\item A3: MRO: $(0, 0, \bm{1})$, DVN: $(0, 0, \bm{1})$, OXY: $(0, 0, \bm{1})$, APA: $(0, 0, \bm{1})$, MOS: $(0, 0.03, \bm{0.97})$.
\end{itemize}

All the five companies in A3 are in the energy industry (oil, mining, etc.), representing the traditional aspect of the financial market. 
Companies in A1 and A2 leave more room for interpretation. 
In particular, for A1, NFLX is an entertainment company, STZ.B is a food company (beer and wine), ILMN is a biotechnology company, REGN is a pharmaceutical company, FLT is a financial service company; for A2, both AMD and LRCX are in the semiconductor industry, FTNT is a cybersecurity company, ISRG is a surgical equipment design company, and CPRT is a company that provides online vehicle auction and remarketing services.
According to the quant ratings on \href{https://seekingalpha.com/}{https://seekingalpha.com/} between 2021 and 2022, all these companies have high profitability; each of these companies has consecutively ranked above A- and many have been A+ in the three latest reports (the factor grade ranges from A+ to F).
This feature is also manifested in the upward trend in the archetypal curves associated with A1 and A2. 
Moreover, they are more resilient than the traditional industries when unexpected events occur (e.g. Covid-19 pandemic in early 2020), as can be seen from the ``V'' shape of these curves near Day $2000$ in the first plot in Figure \ref{fig:44}.
The difference between A1 and A2 is more difficult to corroborate using recent financial data. 
From a macroscopic perspective, A1 represents the more established highly profitable industries in the market; they maintain a steady pace of CLR growth over time.
A2 represents the emerging industries that, while not as profitable as A1, possess relatively more growth potential.
This conclusion can be numerically inspected by comparing the slope of the A1 and A2 curves in Figure \ref{fig:44}.

	\subsection{Intel Image}
	
	The Intel Image dataset \cite{intel} has been used for multi-class classification in machine learning, and consists of $24000$ images representing $6$ different categories of the scene: Buildings, Forest, Glacier, Mountain, Sea, and Street. 
	Each image is a $150\times 150$ pixel color image, which corresponds to a $67500$-dimensional vector through vectorization and stacking of the pixel matrices ($d = 67500$).  
	We randomly select $3000$ samples in the training dataset ($N = 3000$) and apply AAA to extract representative patterns. 
	Note we could have used the full dataset; however, this would require using a more efficient optimization solver for the subproblems to ensure the computation is done in a reasonable time.
	Since there are $6$ different categories of images, we set $k=6$. 
	The input parameters for AAA are chosen as $p = 30$, $M = 10^4$ and $\eta = 0.03$. 
	We compare the computed archetypes given by AAA with the clustering centers given by the $k$-means in Figure \ref{fig:6}. 
	
	\begin{figure}[htbp]
		\begin{center}
			\includegraphics[width=0.16\textwidth]{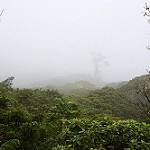}
			\includegraphics[width=0.16\textwidth]{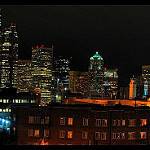}
			\includegraphics[width=0.16\textwidth]{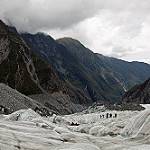}
			\includegraphics[width=0.16\textwidth]{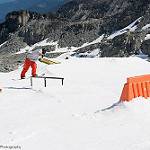}
			\includegraphics[width=0.16\textwidth]{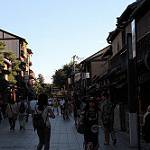}
			\includegraphics[width=0.16\textwidth]{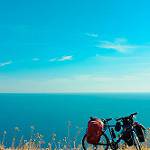}
			\includegraphics[width=0.16\textwidth]{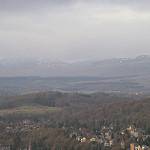}
			\includegraphics[width=0.16\textwidth]{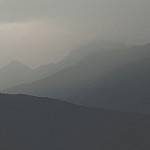}
			\includegraphics[width=0.16\textwidth]{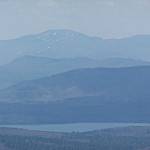}
			\includegraphics[width=0.16\textwidth]{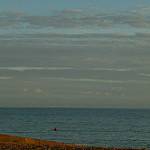}
			\includegraphics[width=0.16\textwidth]{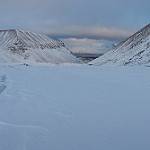}
			\includegraphics[width=0.16\textwidth]{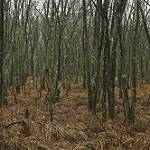}
			\caption{Six images in the dataset with largest component in each archetypal direction identified by AAA (top) compared to the six images in the dataset closest to the centers of the $k$-means (bottom).}
			\label{fig:6}
		\end{center}
	\end{figure}
	
	In this experiment, the instance running time is $527.356$s ($107.147$s for data dimensionality reduction, $1.627$s for representation cardinality reduction and $418.582$s for solving the reduced problem using Algorithm \ref{alg:AM}) for AAA, and $151.787$s for the $k$-means.   
	In this case, the cardinality of the extreme points used to build up the approximate convex hull is $738$. 
	The $6$ archetypes account for about $41.2\%$ of the variance of the dataset, as opposed to $28.3\%$ explained by the $k$-means. 
	The other two methods, SVD-AA, and archetypes cannot be implemented within a reasonable time.  
	
	For each archetype, we find the image that has the largest component with respect to it in the dataset. 
	We also identify the images closest to the $k$-means centers.
	The results are reported in Figure \ref{fig:6}.  
	According to the label information, the images on the top and bottom panels in Figure \ref{fig:6} (from left to right) correspond to ``Forest'', ``Buildings'', ``Glacier'', ``Glacier'', ``Street'', ``Sea'' and ``Mountain'', ``Mountain'', ``Mountain'', ``Sea'', ``Glacier'', ``Forest'', respectively. 
	Despite an approximate algorithm, AAA produces more diversified results than the $k$-means in terms of image content. The only repetition occurs in the third and fourth pictures, where both the snow mountains are classified as Glacier.


	\subsection{MNIST dataset}
	
	The MNIST database \cite{lecun1998mnist} is a large database of handwritten digits that is commonly used for both classification and clustering tasks. 
 	Each data point in MNIST is a $28\times 28$ gray-scale image (i.e. a $784$-dimensional vector) representing handwritten digits from $0$ to $9$. 
	The total size of the training dataset is 42000. 
	
	In this experiment, we use both the $k$-means and AA to analyze the data structure in each label class separately.
	We first split the training data into 10 different datasets corresponding to labeled digits $0, \cdots, 9$, respectively,  each having a size of around 4000. 
	We apply both the $k$-means and AAA to the split datasets to identify the typical patterns and the archetypes, respectively.  
	After running the ``elbow inspection'' for the $k$-means at different labels, we found $k=5$ to be a reasonable choice on average. 
	To be consistent, we also use $k=5$ for AAA. 
	Moreover, the other parameters in AAA are set as $p = 10$, $M = 10^4$ and $\eta = 0.03$. 
	As before, in each label class, we find the images in the corresponding datasets that are closest to the $k$-means centers as well as have the largest convex combination coefficients with respect to the archetypes. 
	The results are reported in Figure \ref{fig:7}.
	
	In general, the $k$-means centers are the images that are representative of each label class.
	They are more standard and usually can be distinguished using raws eyes. 
	On the flip side, the approximate archetypes found by AAA are more extreme in terms of size, shape, position, etc. 
	
	\begin{figure}[htbp]
		\begin{center}
			\includegraphics[width=0.05\textwidth]{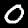}
			\includegraphics[width=0.05\textwidth]{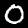}
			\includegraphics[width=0.05\textwidth]{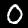}
			\includegraphics[width=0.05\textwidth]{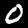}
			\includegraphics[width=0.05\textwidth]{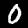}\hspace{3 cm}
			\includegraphics[width=0.05\textwidth]{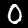}
			\includegraphics[width=0.05\textwidth]{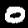}
			\includegraphics[width=0.05\textwidth]{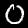}
			\includegraphics[width=0.05\textwidth]{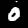}
			\includegraphics[width=0.05\textwidth]{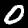}\\
			\includegraphics[width=0.05\textwidth]{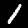}
			\includegraphics[width=0.05\textwidth]{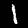}
			\includegraphics[width=0.05\textwidth]{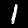}
			\includegraphics[width=0.05\textwidth]{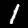}
			\includegraphics[width=0.05\textwidth]{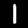}\hspace{3 cm}
			\includegraphics[width=0.05\textwidth]{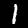}
			\includegraphics[width=0.05\textwidth]{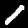}
			\includegraphics[width=0.05\textwidth]{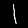}
			\includegraphics[width=0.05\textwidth]{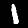}
			\includegraphics[width=0.05\textwidth]{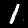}\\
			\includegraphics[width=0.05\textwidth]{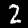}
			\includegraphics[width=0.05\textwidth]{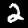}
			\includegraphics[width=0.05\textwidth]{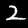}
			\includegraphics[width=0.05\textwidth]{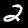}
			\includegraphics[width=0.05\textwidth]{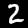}\hspace{3 cm}
			\includegraphics[width=0.05\textwidth]{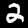}
			\includegraphics[width=0.05\textwidth]{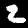}
			\includegraphics[width=0.05\textwidth]{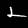}
			\includegraphics[width=0.05\textwidth]{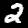}
			\includegraphics[width=0.05\textwidth]{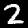}\\
			\includegraphics[width=0.05\textwidth]{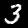}
			\includegraphics[width=0.05\textwidth]{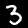}
			\includegraphics[width=0.05\textwidth]{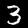}
			\includegraphics[width=0.05\textwidth]{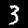}
			\includegraphics[width=0.05\textwidth]{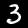}\hspace{3 cm}
			\includegraphics[width=0.05\textwidth]{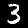}
			\includegraphics[width=0.05\textwidth]{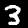}
			\includegraphics[width=0.05\textwidth]{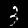}
			\includegraphics[width=0.05\textwidth]{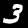}
			\includegraphics[width=0.05\textwidth]{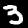}\\
			\includegraphics[width=0.05\textwidth]{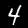}
			\includegraphics[width=0.05\textwidth]{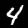}
			\includegraphics[width=0.05\textwidth]{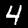}
			\includegraphics[width=0.05\textwidth]{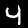}
			\includegraphics[width=0.05\textwidth]{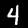}\hspace{3 cm}
			\includegraphics[width=0.05\textwidth]{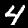}
			\includegraphics[width=0.05\textwidth]{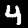}
			\includegraphics[width=0.05\textwidth]{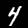}
			\includegraphics[width=0.05\textwidth]{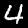}
			\includegraphics[width=0.05\textwidth]{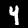}\\
			\includegraphics[width=0.05\textwidth]{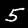}
			\includegraphics[width=0.05\textwidth]{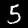}
			\includegraphics[width=0.05\textwidth]{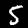}
			\includegraphics[width=0.05\textwidth]{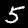}
			\includegraphics[width=0.05\textwidth]{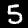}\hspace{3 cm}
			\includegraphics[width=0.05\textwidth]{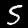}
			\includegraphics[width=0.05\textwidth]{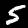}
			\includegraphics[width=0.05\textwidth]{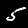}
			\includegraphics[width=0.05\textwidth]{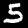}
			\includegraphics[width=0.05\textwidth]{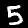}\\
			\includegraphics[width=0.05\textwidth]{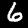}
			\includegraphics[width=0.05\textwidth]{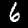}
			\includegraphics[width=0.05\textwidth]{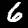}
			\includegraphics[width=0.05\textwidth]{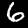}
			\includegraphics[width=0.05\textwidth]{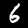}\hspace{3 cm}
			\includegraphics[width=0.05\textwidth]{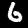}
			\includegraphics[width=0.05\textwidth]{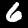}
			\includegraphics[width=0.05\textwidth]{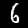}
			\includegraphics[width=0.05\textwidth]{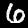}
			\includegraphics[width=0.05\textwidth]{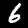}\\
			\includegraphics[width=0.05\textwidth]{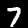}
			\includegraphics[width=0.05\textwidth]{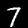}
			\includegraphics[width=0.05\textwidth]{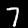}
			\includegraphics[width=0.05\textwidth]{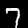}
			\includegraphics[width=0.05\textwidth]{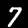}\hspace{3 cm}
			\includegraphics[width=0.05\textwidth]{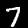}
			\includegraphics[width=0.05\textwidth]{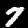}
			\includegraphics[width=0.05\textwidth]{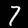}
			\includegraphics[width=0.05\textwidth]{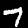}
			\includegraphics[width=0.05\textwidth]{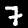}\\
			\includegraphics[width=0.05\textwidth]{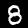}
			\includegraphics[width=0.05\textwidth]{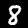}
			\includegraphics[width=0.05\textwidth]{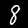}
			\includegraphics[width=0.05\textwidth]{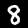}
			\includegraphics[width=0.05\textwidth]{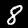}\hspace{3 cm}
			\includegraphics[width=0.05\textwidth]{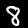}
			\includegraphics[width=0.05\textwidth]{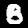}
			\includegraphics[width=0.05\textwidth]{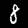}
			\includegraphics[width=0.05\textwidth]{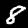}
			\includegraphics[width=0.05\textwidth]{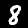}\\
			\includegraphics[width=0.05\textwidth]{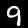}
			\includegraphics[width=0.05\textwidth]{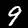}
			\includegraphics[width=0.05\textwidth]{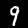}
			\includegraphics[width=0.05\textwidth]{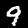}
			\includegraphics[width=0.05\textwidth]{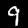}\hspace{3 cm}
			\includegraphics[width=0.05\textwidth]{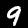}
			\includegraphics[width=0.05\textwidth]{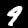}
			\includegraphics[width=0.05\textwidth]{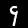}
			\includegraphics[width=0.05\textwidth]{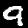}
			\includegraphics[width=0.05\textwidth]{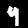}\\
			\caption{Typical patterns (left) and archetypes (right) idenfied by the $k$-means and AAA in each label class in the MNIST training dataset.}
			\label{fig:7}
		\end{center}
	\end{figure}

	\section*{Acknowledgements}
	We would like to thank the anonymous referees for their very helpful comments which significantly improved the presentation of the paper.
	We would like to thank Yu Zhu for providing us with the S\&P 500 dataset and helping clarify some related questions related to interpretation.  
	We also thank Akil Narayan for reading through an early version of the draft, and for providing several comments that improved the presentation of the manuscript.
	Y. Xu would like to thank the organizers of the MSRI Summer Graduate School on Mathematics of Big Data: Sketching and (Multi-) Linear Algebra for motivating discussions.

	\section*{Funding}
	R. Han is supported by the Direct Grant for Research from The Chinese University of Hong Kong, Hong Kong under Grant No. 4053474.
	B. Osting is supported by the National Science Foundation under Grant No. DMS-1752202.  
	D. Wang is supported by the National Natural Science Foundation of China grant 12101524 and the University Development Fund from The Chinese University of Hong Kong, Shenzhen under Grant No. UDF01001803. 
	Y. Xu is supported by the National Science Foundation under Grant No. DMS-1848508.

	\printbibliography

@inproceedings{thorndike1953belongs,
  title={Who belongs in the family},
  author={Thorndike, Robert L},
  booktitle={Psychometrika},
  year={1953},
}

@article{lecun1998mnist,
  title={The MNIST database of handwritten digits},
  author={LeCun, Yann},
  journal={http://yann. lecun. com/exdb/mnist/},
  year={1998}
}

@book{knuth1997art,
  title={The art of computer programming},
  author={Knuth, Donald Ervin},
  volume={3},
  year={1997},
  publisher={Pearson Education}
}

@article{damle2017geometric,
  title={A geometric approach to archetypal analysis and nonnegative matrix factorization},
  author={Damle, Anil and Sun, Yuekai},
  journal={Technometrics},
  volume={59},
  number={3},
  pages={361--370},
  year={2017},
  publisher={Taylor \& Francis}
}

@article{javadi2020nonnegative,
  title={Nonnegative matrix factorization via archetypal analysis},
  author={Javadi, Hamid and Montanari, Andrea},
  journal={J. Amer. Statist. Assoc.},
  volume={115},
  number={530},
  pages={896--907},
  year={2020},
  publisher={Taylor \& Francis}
}

@article{hoeffding1994probability,
  title={Probability Inequalities for Sums of Bounded Random Variables},
  author={Hoeffding, Wassily},
  journal={J. Amer. Statist. Assoc.},
  volume={58},
  number={301},
  pages={13--30},
  year={1963}
}

@article{eckart1936approximation,
  title={The approximation of one matrix by another of lower rank},
  author={Eckart, Carl and Young, Gale},
  journal={Psychometrika},
  volume={1},
  number={3},
  pages={211--218},
  year={1936},
  publisher={Springer}
}

@book{hastie01statisticallearning,
  address = {New York, NY, USA},
  author = {Hastie, Trevor and Tibshirani, Robert and Friedman, Jerome},
  publisher = {Springer New York Inc.},
  series = {Springer Series in Statistics},
  title = {The Elements of Statistical Learning},
  year = 2001
}

@article{rudelson2008littlewood,
  title={The Littlewood--Offord problem and invertibility of random matrices},
  author={Rudelson, Mark and Vershynin, Roman},
  journal={Adv. Math.},
  volume={218},
  number={2},
  pages={600--633},
  year={2008},
  publisher={Elsevier}
}

@article{musco2015randomized,
  title={Randomized block Krylov methods for stronger and faster approximate singular value decomposition},
  author={Musco, Cameron and Musco, Christopher},
  journal={Advances in Neural Information Processing Systems},
  volume={2015},
  pages={1396--1404},
  year={2015}
}

@article{erichson2018randomized,
  title={Randomized nonnegative matrix factorization},
  author={Erichson, N Benjamin and Mendible, Ariana and Wihlborn, Sophie and Kutz, J Nathan},
  journal={Pattern Recognition Letters},
  volume={104},
  pages={1--7},
  year={2018},
  publisher={Elsevier}
}

@misc{intel,
  author = {Intel},
  title = {Intel image classification challenge},
  note = {Public dataset available at 
          \url{https://www.kaggle.com/puneet6060/intel-image-classification}},
}

@inproceedings{qian2018dsanls,
  title={Dsanls: Accelerating distributed nonnegative matrix factorization via sketching},
  author={Qian, Yuqiu and Tan, Conghui and Mamoulis, Nikos and Cheung, David W},
  booktitle={Proceedings of the Eleventh ACM International Conference on Web Search and Data Mining},
  pages={450--458},
  year={2018}
}

@article{halko2009finding,
  title={Finding structure with randomness: Probabilistic algorithms for constructing approximate matrix decompositions},
  author={Halko, Nathan and Martinsson, Per-Gunnar and Tropp, Joel A},
  journal={SIAM Rev.},
  volume={53},
  number={2},
  pages={217--288},
  year={2011},
  publisher={SIAM}
}

@inproceedings{sarlos2006improved,
  title={Improved approximation algorithms for large matrices via random projections},
  author={Sarlos, Tamas},
  booktitle={2006 47th Annual IEEE Symposium on Foundations of Computer Science (FOCS'06)},
  pages={143--152},
  year={2006},
  organization={IEEE}
}

@article{clarkson2017low,
  title={Low-rank approximation and regression in input sparsity time},
  author={Clarkson, Kenneth L and Woodruff, David P},
  journal={J. ACM},
  volume={63},
  number={6},
  pages={1--45},
  year={2017},
  publisher={ACM New York, NY, USA}
}

@article{morup2012archetypal,
  title={Archetypal analysis for machine learning and data mining},
  author={M{\o}rup, Morten and Hansen, Lars Kai},
  journal={Neurocomputing},
  volume={80},
  pages={54--63},
  year={2012},
  publisher={Elsevier}
}

@inproceedings{abrol2020geometric,
  title={A geometric approach to archetypal analysis via sparse projections},
  author={Abrol, Vinayak and Sharma, Pulkit},
  booktitle={International Conference on Machine Learning},
  pages={42--51},
  year={2020},
  organization={PMLR}
}

@article{thurau2011convex,
  title={Convex non-negative matrix factorization for massive datasets},
  author={Thurau, Christian and Kersting, Kristian and Wahabzada, Mirwaes and Bauckhage, Christian},
  journal={Knowledge and information systems},
  volume={29},
  number={2},
  pages={457--478},
  year={2011},
  publisher={Springer}
}

@article{mair2019coresets,
  title={Coresets for Archetypal Analysis},
  author={Mair, Sebastian and Brefeld, Ulf},
  journal={Advances in Neural Information Processing Systems},
  volume={32},
  pages={7247--7255},
  year={2019}
}

@inproceedings{mei2018online,
  title={Online dictionary learning for approximate archetypal analysis},
  author={Mei, Jieru and Wang, Chunyu and Zeng, Wenjun},
  booktitle={Proceedings of the European Conference on Computer Vision (ECCV)},
  pages={486--501},
  year={2018}
}

@inproceedings{mair2017frame,
  title={Frame-based data factorizations},
  author={Mair, Sebastian and Boubekki, Ahcene and Brefeld, Ulf},
  booktitle={International Conference on Machine Learning},
  pages={2305--2313},
  year={2017},
  organization={PMLR}
}

@inproceedings{bauckhage2015archetypal,
  title={Archetypal analysis as an autoencoder},
  author={Bauckhage, Christian and Kersting, Kristian and Hoppe, Florian and Thurau, Christian},
  booktitle={Workshop New Challenges in Neural Computation},
  pages={8},
  year={2015},
  organization={Citeseer}
}

@article{tropp2017practical,
  title={Practical sketching algorithms for low-rank matrix approximation},
  author={Tropp, Joel A and Yurtsever, Alp and Udell, Madeleine and Cevher, Volkan},
  journal={SIAM J. Matrix Anal. Appl.},
  volume={38},
  number={4},
  pages={1454--1485},
  year={2017},
  publisher={SIAM}
}

@article{battaglino2018practical,
  title={A practical randomized CP tensor decomposition},
  author={Battaglino, Casey and Ballard, Grey and Kolda, Tamara G},
  journal={SIAM J. Matrix Anal. Appl.},
  volume={39},
  number={2},
  pages={876--901},
  year={2018},
  publisher={SIAM}
}

@article{wang2015fast,
  title={Fast and guaranteed tensor decomposition via sketching},
  author={Wang, Yining and Tung, Hsiao-Yu and Smola, Alexander J and Anandkumar, Anima},
  journal={Advances in neural information processing systems},
  volume={28},
  year={2015}
}

@inproceedings{makarychev2019performance,
  title={Performance of Johnson-Lindenstrauss transform for k-means and k-medians clustering},
  author={Makarychev, Konstantin and Makarychev, Yury and Razenshteyn, Ilya},
  booktitle={Proceedings of the 51st Annual ACM SIGACT Symposium on Theory of Computing},
  pages={1027--1038},
  year={2019}
}

@inproceedings{cohen2015dimensionality,
  title={Dimensionality reduction for k-means clustering and low rank approximation},
  author={Cohen, Michael B and Elder, Sam and Musco, Cameron and Musco, Christopher and Persu, Madalina},
  booktitle={Proceedings of the forty-seventh annual ACM symposium on Theory of computing},
  pages={163--172},
  year={2015}
}

@inproceedings{cohen2017input,
  title={Input sparsity time low-rank approximation via ridge leverage score sampling},
  author={Cohen, Michael B and Musco, Cameron and Musco, Christopher},
  booktitle={Proceedings of the Twenty-Eighth Annual ACM-SIAM Symposium on Discrete Algorithms},
  pages={1758--1777},
  year={2017},
  organization={SIAM}
}

@article{boutsidis2010random,
  title={Random Projections for $ k $-means Clustering},
  author={Boutsidis, Christos and Zouzias, Anastasios and Drineas, Petros},
  journal={Advances in Neural Information Processing Systems},
  volume={23},
  pages={298--306},
  year={2010}
}

@article{avron2010blendenpik,
  title={Blendenpik: Supercharging LAPACK's least-squares solver},
  author={Avron, Haim and Maymounkov, Petar and Toledo, Sivan},
  journal={SIAM J. Sci. Comput.},
  volume={32},
  number={3},
  pages={1217--1236},
  year={2010},
  publisher={SIAM}
}

@article{drineas2011faster,
  title={Faster least squares approximation},
  author={Drineas, Petros and Mahoney, Michael W and Muthukrishnan, Shan and Sarl{\'o}s, Tam{\'a}s},
  journal={Numer. Math.},
  volume={117},
  number={2},
  pages={219--249},
  year={2011},
  publisher={Springer}
}

@article{shoval2012evolutionary,
  title={Evolutionary trade-offs, Pareto optimality, and the geometry of phenotype space},
  author={Shoval, Oren and Sheftel, Hila and Shinar, Guy and Hart, Yuval and Ramote, Omer and Mayo, Avi and Dekel, Erez and Kavanagh, Kathryn and Alon, Uri},
  journal={Science},
  volume={336},
  number={6085},
  pages={1157--1160},
  year={2012},
  publisher={American Association for the Advancement of Science}
}

@article{graham2017approximate,
  title={Approximate Convex Hulls: sketching the convex hull using curvature},
  author={Graham, Robert and Oberman, Adam M},
  journal={arXiv preprint arXiv:1703.01350},
  year={2017}
}

@article{woodruff2014sketching,
  title={Sketching as a Tool for Numerical Linear Algebra},
  author={Woodruff, David P},
  journal={Foundations and Trends{\textregistered} in Theoretical Computer Science},
  volume={10},
  number={1--2},
  pages={1--157},
  year={2014},
  publisher={Now Publishers Inc. Hanover, MA, USA}
}

@article{cutler1994archetypal,
  title={Archetypal analysis},
  author={Cutler, Adele and Breiman, Leo},
  journal={Technometrics},
  volume={36},
  number={4},
  pages={338--347},
  year={1994},
  publisher={Taylor \& Francis}
}

@book{vershynin2018high,
  title={High-dimensional probability: An introduction with applications in data science},
  author={Vershynin, Roman},
  volume={47},
  year={2018},
  publisher={Cambridge university press}
}

@article{osting2021consistency,
  title={Consistency of archetypal analysis},
  author={Osting, Braxton and Wang, Dong and Xu, Yiming and Zosso, Dominique},
  journal={SIAM J. Math. Data Sci.},
  volume={3},
  number={1},
  pages={1--30},
  year={2021},
  publisher={SIAM}
}

@inproceedings{kozlov1979polynomial,
  title={Polynomial solvability of convex quadratic programming},
  author={Kozlov, Mikhail K and Tarasov, Sergei Pavlovich and Khachiyan, Leonid Genrikhovich},
  booktitle={Doklady Akademii Nauk},
  volume={248},
  number={5},
  pages={1049--1051},
  year={1979},
  organization={Russian Academy of Sciences}
}

@article{lax2002functional,
  title={Functional Analysis. John Wiley\&Sons},
  author={Lax, Peter D},
  journal={Inc. Publication},
  year={2002}
}

@manual{quadprog,
  author = {Turlach, Berwin A. and Weingessel, Andreas},
  note = {R package version 1.5-4},
  title = {quadprog: {Functions} to solve Quadratic Programming Problems.},
  url = {http://CRAN.R-project.org/package=quadprog},
  year = 2011
}

@Manual{R,
    title = {R: A Language and Environment for Statistical Computing},
    author = {{R Core Team}},
    organization = {R Foundation for Statistical Computing},
    address = {Vienna, Austria},
    year = {2020},
    url = {https://www.R-project.org/},
  }

@Article{aa1,
    title = {From {S}pider-{M}an to {H}ero -- Archetypal Analsis in
      {R}},
    author = {Manuel J. A. Eugster and Friedrich Leisch},
    journal = {Journal of Statistical Software},
    year = {2009},
    volume = {30},
    number = {8},
    pages = {1--23},
    url = {http://www.jstatsoft.org/v30/i08/},
  }

@Article{aa2,
    title = {Weighted and Robust Archetypal Analysis},
    author = {Manuel J. A. Eugster and Friedrich Leisch},
    journal = {Comput. Statist. Data Anal.},
    year = {2011},
    volume = {55},
    number = {3},
    pages = {1215--1225},
    url =
      {http://www.sciencedirect.com/science/article/pii/S0167947310004056},
    preprint = {http://epub.ub.uni-muenchen.de/11498/},
  }

@inproceedings{chen2014fast,
  title={Fast and robust archetypal analysis for representation learning},
  author={Chen, Yuansi and Mairal, Julien and Harchaoui, Zaid},
  booktitle={Proceedings of the IEEE Conference on Computer Vision and Pattern Recognition},
  pages={1478--1485},
  year={2014}
}

\end{document}